\documentclass{JHEP3}

\usepackage{url,amsmath,amsthm,amscd,graphicx,psfrag}
\usepackage[all]{xy}
\numberwithin{equation}{section}
\numberwithin{table}{section}

\def\2{{1\over2}}
\let\<=\langle \let\>=\rangle
\def\new#1\endnew{{\bf #1}}
\def\ifundefined#1{\expandafter\ifx\csname#1\endcsname\relax}
\ifundefined{draftmode}\else    \input draftmode        \fi
\let\Msize=\footnotesize             
\def\BM{\Msize\begin{matrix}}           \def\EM{\end{matrix}}
\def\MN M:#1 #2 N:#3 #4 {{(#1_{#2},#3_{#4})}}
\def\MNH M:#1 #2 N:#3 #4 H:#5,#6 [#7]{{(#1_{#2},#3_{#4})^{#5,#6}_{#7}}}

\def\dd{\mathrm{d}}

\def\CF{{\cal F}}

\def\CM{{\cal M}}

\def\CO{{\cal O}}

\def\CW{{\cal W}}

\newcommand{\be}{\begin{equation}}
\newcommand{\ee}{\end{equation}}
\newcommand{\bea}{\begin{eqnarray}}
\newcommand{\eea}{\end{eqnarray}}

\def\IZ{{\mathbb Z}}
\def\IR{{\mathbb R}}
\def\IC{{\mathbb C}}
\def\IP{{\mathbb P}}

\newcommand{\re}{{\rm e}}
\newcommand{\ri}{{\rm i}}
\newcommand{\rd}{{\rm d}}

\newcommand{\ba}{\begin{aligned}}
\newcommand{\ea}{\end{aligned}}
\newcommand{\ben}{\begin{eqnarray}\displaystyle}
\newcommand{\een}{\end{eqnarray}}

\def\Im{\text{Im}}

\theoremstyle{definition}
\newtheorem{defn}{Definition}[section]
\newtheorem{exmp}[defn]{Example}

\theoremstyle{plain}
\newtheorem{prop}[defn]{Proposition}

\newtheorem{lem}[defn]{Lemma}
\newtheorem{conj}[defn]{Conjecture}

\theoremstyle{remark}
\newtheorem{rem}[defn]{Remark}

\title{Topological open strings on orbifolds}

\author{Vincent Bouchard\\
Jefferson Physical Laboratory, Harvard University\\
17 Oxford St., Cambridge, MA 02138, USA\\
E-mail: \email{bouchard@physics.harvard.edu}}

\author{Albrecht Klemm\\
Physikalisches Institut der Universit\"at Bonn\\ 
Nu\ss allee 12, D-53115 Bonn, Germany\\
E-mail: \email{aklemm@th.physik.uni-bonn.de}}

\author{Marcos Mari\~no\\
D\'epartement de Physique Th\'eorique et Section de Math\'ematiques\\
Universit\'e de Gen\`eve, CH 1211 Gen\`eve 4, Suisse\\
E-mail: \email{marcos.marino@physics.unige.ch}}

\author{Sara Pasquetti\\
Institut de Physique, Universit\'e de Neuch\^atel\\
Rue A. L. Breguet 1, CH-2000 Neuch\^atel, Suisse\\
E-mail: \email{sara.pasquetti@unine.ch}}

\abstract{We use the remodeling approach to the B-model topological string in terms of recursion relations 
to study open string amplitudes at orbifold points. To this end, we clarify modular properties 
of the open amplitudes and rewrite them in a form that makes their transformation 
properties under the modular group manifest. We exemplify this procedure for the $\IC^3/\IZ_3$ orbifold 
point of local $\IP^2$, where we present results for topological string amplitudes for genus zero and up to three holes, 
and for the one-holed torus. These amplitudes can be understood as generating functions for either open orbifold Gromov--Witten invariants of $\IC^3 / \IZ_3$, or correlation functions in the orbifold CFT involving 
insertions of both bulk and boundary operators.}

\begin{document}

\section{Introduction}

%\subsection{Background}

Topological string theory on Calabi--Yau threefolds has played a crucial role in our understanding of string theory and 
 Gromov--Witten theory. One of the most fascinating aspects of this topological sector of string theory is that 
very often amplitudes can be computed exactly, and their dependence on the moduli can be studied in detail. This has 
led to very rich pictures of the moduli space of the theory, involving
different phases which exhibit different physics \cite{wittenphases,agm}. 

Modular and analytic properties of the amplitudes connect the different phases of the Calabi--Yau moduli space in a
very precise way. Each phase of the moduli space is characterized by a 
set of ``good coordinates,'' and different good coordinates corresponding to 
different phases are related by a transformation in the modular 
group of the theory. As explained in \cite{ABK}, topological string amplitudes are modular objects with specific 
transformation properties under this group, and as one goes from one phase to the other, 
the amplitudes have to be transformed accordingly. For example, when expanded
at the large radius limit in moduli space, topological string 
amplitudes are generating functions of Gromov--Witten invariants. As one moves
away from this point towards different regions 
in moduli space, the large radius expansion eventually ceases to converge, 
but after suitable modular transformations and analytic continuations, 
the topological string amplitudes can be re-expanded in terms of the 
good variables of the new phase. In particular, when going to orbifold points of the moduli space, the amplitudes become generating functions 
for {\it orbifold} Gromov--Witten invariants. A detailed understanding of the
modular transformation properties of the amplitudes makes it 
then possible to relate Gromov--Witten invariants to orbifold Gromov--Witten
invariants, in the spirit of the crepant resolution conjecture \cite{BG,CCIT,Ru}. In \cite{ABK} this was used to calculate generating functions of 
orbifold Gromov--Witten invariants in the case of the $\IC^3/\IZ_3$ orbifold, 
which corresponds to a phase in the moduli space of local $\IP^2$, its crepant resolution.  The predictions obtained in this way have
been later verified mathematically in orbifold Gromov--Witten 
theory~\cite{BaC,BC,CC,CCIT2}, and other examples have been recently calculated \cite{BrT,Coates}. 

A crucial ingredient in the approach of \cite{ABK} is the ability to obtain exact results for the topological string 
amplitudes on the whole of moduli space, so that they can be expanded in
different phases. These exact expressions are typically calculated 
by using the B-model and mirror symmetry. On top of that, it is extremely
useful to write these exact results in a way that makes the transformation
properties manifest. For local Calabi--Yau threefolds, the mirror manifold
reduces to an algebraic curve and the modular group is 
essentially the symplectic group acting on the homology of the surface. 
Topological string amplitudes can then be written in terms of modular forms
with respect to this group, and when the curve has genus one, 
as in the case of the mirror to local $\IP^2$, one can write them in terms of elliptic functions 
\cite{ABK}. 

The results of \cite{ABK} were obtained for {\it closed} string amplitudes, and it is natural to ask how one could extend these results to {\it open} topological string amplitudes. 
As in the closed case, we first need a formalism to
compute open topological string amplitudes exactly on the whole closed and
open moduli space. For the case of toric Calabi--Yau threefolds, this
formalism has been proposed in \cite{MM,BKMP} and it is based on a recursion
relation first obtained in the context of matrix models \cite{eynard, EO}. One
advantage of the framework developed in \cite{MM,BKMP}, as compared to the holomorphic 
anomaly equations of \cite{BCOV,Wa}, is that the amplitudes are completely
fixed by the recursion. It is then natural to use this formalism in order to 
understand the properties of open string amplitudes as one moves in the open
and closed moduli space of toric Calabi--Yau threefolds, 
and in particular to extract information about the open counterparts of 
orbifold Gromov--Witten invariants (which so far have not been defined 
in the mathematical literature). 

In \cite{BKMP} some steps were taken in this direction. In particular we
discussed how to find ``good coordinates'' for the open moduli at the orbifold point, and 
we made a preliminary analysis of the disk amplitude. In this paper we present
a detailed study of open topological string amplitudes at the orbifold point, focusing 
on the case of $\IC^3 / \IZ_3$. First of all, we clarify the transformation
properties of the string amplitudes in the open sector, and we present expressions 
for them which make  their modular transformation properties manifest. 
Since the recursion of \cite{EO,MM,BKMP} is based on the Bergman kernel of the mirror curve, 
our first step is to write it (for a curve of genus one) in terms of elliptic
functions. One can then plug the resulting expression in the recursion to find 
modular expressions for all the open string amplitudes. This leads to considerable improvements in terms of computional efficiency of 
the recursion relations. As a consequence
of this refinement of the formalism of \cite{BKMP}, we are able to calculate
open orbifold string amplitudes at high order, and we present 
explicit expressions for amplitudes with $(g,h)=(0,2)$, $(0,3)$ and $(1,1)$. 
These expressions give generating functions for open orbifold Gromov--Witten
invariants, and from the CFT point of view they compute correlation functions
of arbitrary insertions of both bulk operators, associated with twist fields,
and boundary operators, associated with deformation modes of the D-brane open moduli. 

The organization of the paper is as follows. We start by reviewing the remodeling approach to the B-model using recursion relations in sections 2.1 to 2.3. In section 2.4, we study modularity of the open amplitudes, which we rewrite in a form that makes their transformation properties explicit in section 2.5. Section 3 is then devoted to the study of topological open string amplitudes at the $\IC^3 / \IZ_3$ orbifold point in the moduli space of local $\IP^2$, using the formalism presented in section 2. We also briefly comment on the calculation of the open amplitudes at the conifold point in the moduli space of local $\IP^2$ in section 3.5.

\subsection*{Acknowledgments}
We would like to thank Alireza Tavanfar and Marlene Weiss for collaboration at the initial stages of this work, and Renzo Cavalieri for sharing with us his results on open orbifold Gromov--Witten invariants prior to publication. 
We would also like to thank Paul Johnson, Manfred Herbst, Yongbin Ruan and Ed Segal for interesting discussions.
The work of S.P. was partly supported by the Swiss National Science Foundation and by the European Commission under contracts MRTN-CT-2004-005104.
The work of V.B. is supported in part by the Center for the Fundamental Laws of Nature at Harvard University and by NSF grants PHY-0244821 and DMS-0244464.

\section{Open B-model on mirrors of toric Calabi-Yau threefolds}

\subsection{The geometry}

Consider the A-twisted sigma model on a (noncompact) toric Calabi-Yau
threefold $X$. A-branes are objects in the ``derived Fukaya category'' of $X$;
roughly speaking, they correspond to Lagrangian submanifolds of $X$ with
bundles on them. We consider a simple class of A-branes, given by noncompact
special Lagrangian submanifolds $L \subset X$ with trivial bundle, 
with topology $\IR^2 \times S^1$; those were constructed in
\cite{HarveyLawson,AV,AKV} --- see also \cite{BKMP} for a detailed description.

The mirror theory is a B-twisted sigma model\footnote{The mirror is generally
presented as a Landau-Ginzburg model; we explain the correspondence between
the Landau-Ginzburg model and the sigma model in Appendix A.} on a family
$\pi: Y \to \CM$ of noncompact Calabi-Yau threefolds, where $\CM$ is the
moduli space of the closed B-model. 
Let $z=(z_1,\ldots,z_k)$ be coordinates on $\CM$ centered at a point of
maximally unipotent monodromy. The fiber $Y_z = \pi^{-1}(z_1,\ldots,z_k)$ of the family has the form
\be \label{e:Y_z}
Y_z = \{ w w' = H(x,y;z) \} \subset (\IC)^2 \times (\IC^*)^2,
\ee
where $H(x,y;z)$ is a Laurent polynomial in $x,y \in \IC^*$ of degree 1. The
precise form of $H(x,y;z)$ is dictated 
by the toric data of the mirror $X$. $Y_z$ is a quadric fibration over $(\IC^*)^2$, with degeneration locus the Riemann surface
\be\label{e:sigmazfirst}
\Sigma_z = \{ H(x,y;z)=0 \} \subset (\IC^*)^2.
\ee

B-branes are objects in the derived category of coherent sheaves, some of which correspond to holomorphic submanifolds of $Y_z$ with bundles on them. The B-branes mirror to the simple A-branes considered above can be described as wrapping a holomorphic curve in $Y_z$, with trivial bundle on it. More precisely, fix a point $p_0 \in \Sigma_z$ parameterized by $(x_0,y_0)$, and denote by $C_z(p_0)$ the holomorphic submanifold of $Y_z$ defined by
\be
w' = 0 = H(x_0,y_0;z).
\ee
It is given by the line parameterized by $w$ over the point $p_0 \in \Sigma_z$. This is the holomorphic curve which is wrapped by the B-brane. The open moduli space corresponds to deformations of the B-brane $C_z(p_0)$ in $Y_z$, which are parameterized by the point $p_0 \in \Sigma_z$. As a result, the moduli space of the open B-model on $(Y,C)$ is given by the family of Riemann surfaces $\Sigma \to \CM$, with fiber \eqref{e:sigmazfirst}.

\begin{exmp}
The main example that we will study is the mirror to local $\IP^2$. Let $X
=K_{\IP^2}$ be the total space of the canonical bundle over $\IP^2$. Its
mirror is the family of Calabi-Yau threefolds $Y \to \CM$, where the closed
moduli space $\CM$ is one-dimensional, whose fibers $Y_z$ are given by
(\ref{e:Y_z}) with  
\be \label{e:defmp2}
H(x,y;z)=1+x+y+\frac{z}{x y}.
\ee
The family of Riemann surfaces $\Sigma \to \CM$ has fibers $\Sigma_z$ (\ref{e:sigmazfirst}),
which are elliptic curves with three punctures.
\label{ex:localp2}
\end{exmp}

\subsection{Disk amplitude}

In this paper we focus on the open amplitudes of the B-model. Let us start with the simplest amplitude, the disk amplitude (genus $0$, $1$ hole). Roughly speaking, it is the open analog of the genus $0$ closed amplitude, which corresponds to the prepotential of special geometry of the closed moduli space $\CM$. The disk amplitude on $(Y,C)$ similarly admits a simple definition as follows.

Recall that the moduli space of the open B-model consists in a family of 
Riemann surfaces (with punctures) $\Sigma \to \CM$. Choose an embedding\footnote{The choice of embedding of $\Sigma_z$ in $(\IC^*)^2$ corresponds to a choice of phase and framing of the mirror brane. This was considered in detail in \cite{BKMP}.} of the fibers $\Sigma_z$ in $(\IC^*)^2$,
\be
\Sigma_z = \{ H(x,y;z)=0 \} \subset (\IC^*)^2,
\ee
and define the one-form
\begin{align}
\omega(p) =& \log y(x(p)) \frac{\dd x(p)}{x(p)}\notag\\
=& \log y(x) \frac{\dd x}{x}
\end{align}
on $\Sigma_z$, where $p \in \Sigma_z$ and $x$ is chosen as local coordinate.

\begin{rem}\label{rem:dep}
Note that in the following we will always omit the dependence on $z$ to simplify the notation. But since $\Sigma \to \CM$ is a family of curves, all the objects we define on the fiber $\Sigma_z$ will have an implicit dependence on $z$.
\end{rem}

The main conjecture of \cite{AKV, AV}, which comes from dimensional reduction of the holomorphic Chern-Simons theory on the brane $C$, goes as follows.

\begin{conj}[\cite{AKV,AV}]
The ``Abel-Jacobi'' map
\be
F^{(0,1)} = \int_\gamma \omega(p),
\ee
where $\gamma$ is the chain $[q^*,q]$ and $q^* \in \Sigma_z$ is a reference point, gives the B-model disk amplitude, up to classical terms. $F^{(0,1)}$ should be understood as a series expansion in the local coordinate $x$ near $x=0$, where $x$ corresponds to the open modulus associated to the brane.
\end{conj}

The Abel-Jacobi map is defined on the Jacobian, that is only up to addition of integrals of $\omega(p)$ over one-cycles. But here we will only be interested in the series expansion of the amplitude in the open modulus, and so the ambiguity is irrelevant. Note that this conjecture is the local analog of the result of \cite{MW}, where the disk amplitude is computed in terms of normal functions.

This formula has been verified in many examples, by expanding the disk amplitude near a point of maximally unipotent monodromy in the closed moduli space, and comparing with open A-model amplitudes on the toric mirror. It requires an explicit knowledge of the closed and open mirror maps, which can be understood as solutions of an extended Picard-Fuchs system (the latter was derived in the language of mixed Hodge structures and relative cohomology in \cite{LMW}).

\subsection{General formalism}

We now move on to the general amplitudes $F^{(g,h)}$ with genus $g$ and $h$ holes. 
As for the closed amplitudes $F^{(g)}$, the physical B-model open amplitudes 
are generally non-holomorphic, and satisfy an open analog of the holomorphic
anomaly equations of \cite{BCOV}. However, to compare with the A-model
Gromov-Witten generating functions, one needs to consider the holomorphic 
limit of the physical B-model amplitudes expanded near a special point in the
moduli space. The $F^{(g,h)}$ that we consider here are these holomorphic
objects, rather than the physical B-model amplitudes. Stated from a modularity 
point of view, what we construct here are the quasi-modular forms, rather than the almost holomorphic modular forms \cite{ABK}. We will discuss this point in more detail in the next subsection.

In \cite{BKMP,MM} a general recursive formalism for computing B-model genus $g$, $h$ hole open amplitudes $F^{(g,h)}$ on $(Y,C)$ was proposed. From a mathematical point of view, since the open B-model is not really well understood, this can be taken as a proposal for a definition of the open B-model on these geometries.

Consider again the following data:
\begin{itemize}
\item A family of (punctured) Riemann surfaces $\Sigma \to \CM$ (the open B-model moduli space);
\item A choice of embedding of the fibers $\Sigma_z$ in $(\IC^*)^2$,
\be
\Sigma_z = \{ H(x,y;z)=0 \} \subset (\IC^*)^2.
\ee
\end{itemize}
We claim that these data fully characterize the open B-model on $(Y,C)$, with
arbitrary genus and number of holes.

By projecting onto the $x$-axis we may see $\Sigma_z$ as a branched cover of
$\IC^*$. Denote by $q_i \in \Sigma_z$ the ramification points of the
projection map, such that $\dd x (q_i) = 0$. Let $\lambda_i := x(q_i) \in
\IC^*$ be the branch points. We assume that they all have branching order
two. Then, near $q_i$, there are two points $q, \bar q \in \Sigma_z$ 
with the same projection $x(q) = x(\bar q)$ (those are defined only 
locally near $q_i$). As before, define the one-form $\omega(p)$, which reads in local coordinates
\be
\omega(p) = \log y(x) \frac{\dd x}{x}.
\ee

\begin{defn}
The \emph{Bergman kernel} $B(p,q)$ is the unique bilinear differential on $\Sigma_z$ with a double pole at $p=q$ with no residue, and no other pole. It is normalized by
\be
\oint_{A^I} B(p,q) = 0,
\ee
where $(A^I,B_I)$ is a symplectic basis of cycles on $\Sigma_z$.
\end{defn}

Note that the Bergman kernel is defined on the Riemann surface itself, and
does not depend on the embedding in $(\IC^*)^2$. Its definition however 
requires a choice of symplectic basis of cycles on $\Sigma_z$.

\begin{defn}
Near $q_i \in \Sigma_z$, define the one-form
\be
\dd E_{q, \bar q}(p) = \frac{1}{2} \int_{\xi =q}^{\bar q} B(p,\xi),
\ee
where the integration is in a neighborhood of $q_i$. Note that this is defined only locally near $q_i$.
\end{defn}

We are now ready to state the recursion, which was first derived in the context of matrix models in \cite{eynard,CEO,EO}.

\begin{defn}\label{def:wgh}
Let $\tilde W^{(g, h)} (p_1, p_2, \ldots, p_{h})$, $g,h \in \IZ$, $g\geq 0$, $h \geq 1$, be multilinear differentials on $\Sigma_z$. Fix the initial conditions
\be
\tilde W^{(0,1)}(p_1)=0, \qquad \tilde W^{(0,2)}(p_1,p_2) = B(p_1,p_2).
\ee
Define the remaining differentials by the recursion\footnote{Note that the integrand in the right-hand side is only defined locally near the ramification points $q_i$.}
\begin{align}
\tilde W^{(g,h)}(p_1, p_2 \ldots, p_{h}) =& \sum_{q_i}  \underset{q=q_i}{\rm Res~} {\dd E_{q, \bar q}(p_1) \over \omega(q) - \omega(\bar q)} \Big ( \tilde W^{(g-1,h+1)} (q, \bar q, p_2, \ldots, p_{h} )\notag\\
& \qquad +\sum_{l=0}^g \sum_{J\subseteq H} \tilde W^{(g-l,|J|+1 )}(q, p_J) \tilde W^{(l,|H|-|J| +1)} (\bar q, p_{H\backslash J}) \Big),
\label{e:rec1}
\end{align}
where we used the notation $H=\{2,3, \ldots, h\}$, and given any subset 
$J=\{i_1, \ldots, i_j\}\subseteq H$ we defined $p_J=\{p_{i_1}, \ldots, p_{i_j}\}$.
\end{defn}

There is a second recursion which reads as follows.

\begin{defn}\label{def:fg}
Let $F^{(g)}$, $g \in \IZ$, $g \geq 2$, be functions on $\Sigma_z$ defined by
\be
F^{(g)} = \frac{1}{2g-2} \sum_{q_i}  \underset{q=q_i}{\rm Res~} \theta(q) \tilde W^{(g,1)}(q),
\label{e:rec2}
\ee
where $\theta(q)$ is any primitive of $\omega(q)$, \emph{i.e.} $\dd \theta(q) = \omega(q)$.
\end{defn}

We define the function $F^{(1)}$ by:

\begin{defn}\label{def:f1}
Define
\be\label{e:deff1}
F^{(1)} = -{1\over 2} \log \tau_B - \frac{1}{24} \log \prod_{i} \omega'(q_i),
\ee
where $\tau_B$ is the Bergman tau-function and 
\be
\omega'(q_i) ={1\over \rd z_i(p)} \rd \left( {\log y(x)\over x}\right)\biggl|_{p=q_i}, \qquad z_i(p) ={\sqrt {x(p)-x(q_i)}}.
\ee
We refer the reader to \cite{EO} for more details.
\end{defn}

The main conjecture of \cite{BKMP,MM}, which relates the objects defined above recursively to the B-model amplitudes, could be stated as follows.

\begin{conj}
Let $F^{(0)}$ be the prepotential of special geometry, $F^{(1)}$ be as in definition \ref{def:f1}, and the $F^{(g)}$'s for $g \geq 2$ be as in definition \ref{def:fg}. 

For $g \geq 0$, $h \geq 1$, and $(g,h) \neq (0,1), (0,2)$, define the multilinear differentials 
\be
W^{(g,h)}(p_1, \ldots, p_h) = \tilde W^{(g,h)}(p_1, \ldots, p_h),
\ee
using definition \ref{def:wgh}. Let 
\be
W^{(0,2)}(p_1,p_2) = B(p_1,p_2) - \frac{\dd p_1 \dd p_2}{(p_1-p_2)^2},
\ee
and 
\be
W^{(0,1)}(p) = \omega(p).
\ee 
Define
\be
F^{(g,h)} = \int_{\gamma_1} \cdots \int_{\gamma_h} W^{(g,h)}(p_1, \ldots, p_h),
\ee
where the $\gamma_i$'s are the chains $[q_i^*, q_i]$, with the $q_i^* \in \Sigma_z$ reference points. 

The $F^{(g)}$ constructed above are the genus $g$ closed B-model amplitudes on $Y$, and
the $F^{(g,h)}$ are the genus $g$, $h$ hole open B-model amplitudes on $(Y,C)$. The $F^{(g,h)}$ should be understood as series expansions in the local coordinates $x_i := x(p_i)$, which correspond to the open moduli associated to the branes.
\end{conj}

Note that as for the disk amplitude $F^{(0,1)}$, the $F^{(g,h)}$ are only defined modulo integration over closed cycles; but again, this ambiguity will be irrelevant since we only consider instanton expansions of the amplitudes.

\begin{rem}
Note that the conjecture for $F^{(1)}$ should probably follow from the result of Dubrovin and Zhang for the $G$-function associated to Frobenius manifolds \cite{DZ}, which was also studied by Givental \cite{Gi}. It can also be understood from a topological field theory point of view as in \cite{BCOV2}.
\end{rem}

There are various arguments behind this conjecture. First, a strong piece of evidence comes from direct calculation. In \cite{BKMP,MM}, various amplitudes for the mirrors of $\IC^3$, local $\IP^1$, $\IP^2$, $\IP^1 \times \IP^1$, $F_1$, $F_2$ were computed. By expanding the amplitudes near a point of maximally unipotent monodromy and plugging in the open and closed mirror maps, it was shown that one recovers the open A-model amplitudes on the toric mirrors. This however only tests the conjecture at large radius; in \cite{BKMP} the conjecture was also tested at the orbifold point of local $\IP^1 \times \IP^1$, by comparing with perturbative Chern-Simons theory on the lens space $S^3 / \IZ_2$.

A more conceptual argument for the conjecture goes as follows. The recursions were derived by \cite{CEO, EO} in the context of matrix models. When $\Sigma$ is the spectral curve of a matrix model, the recursions \eqref{e:rec1} and \eqref{e:rec2} respectively generate the correlation functions and free energies of the matrix model. For some B-model geometries, using large $N$ dualities on the mirror side, we can find matrix model representations with spectral curve $\Sigma_z$, which justifies the conjecture. However, in general no matrix model representation is known; but it was argued first in \cite{MM}, and then in much more detail in \cite{DV}, that the B-model amplitudes should indeed satisfy the recursions \eqref{e:rec1} and \eqref{e:rec2}, whether there is a matrix model representation or not. This involves a detailed analysis of the B-model, understood in the chiral boson picture developed in \cite{ADKMV}. One can also show that the amplitudes obtained through the recursion, after restoring non-holomorphicity using modularity as in \cite{ABK}, satisfy the holomorphic anomaly equations (and their open analogs) \cite{EMO}. In any case, in the following we will take this conjecture for granted and explore some of its consequences.

\subsection{Modularity}

In the previous subsection we introduced a recursive formalism to compute open
and closed B-model amplitudes. While the formalism is very elegant
conceptually, and provides a complete solution to the B-model on these
geometries, it turns out to be rather complicated computationally. One reason
is that the formalism makes no explicit use of the modular properties of the
amplitudes; on the contrary, the intermediate step of taking residues at the
branch points destroys the symmetry of the amplitudes. Indeed, the branch
points are in general complicated functions of $z$, since the projection 
$\Sigma_z \to \IC^*$ is a branched cover. But the final amplitudes are simple 
functions of $z$; in fact, only symmetric combinations of the branch points, 
which are simple rational functions of $z$, appear in the final amplitudes. 

It thus seems desirable to recast the recursion in a different form, bypassing
the intermediate step of taking residues at the branch points. However, one
problem is that the integrand in the right hand side of \eqref{e:rec1} is only
defined locally near the branch points. Hence, one cannot simply deform the
contour integral to pick up residues at the other poles in a straightforward
way. Indeed, localizing the integrand at the branch points turns out to be
crucial in the derivation of the recursion in the matrix model context 
(see for instance pp.14-15 of \cite{CEO}), in order to get rid of unfixed
polynomials. So it seems that the intermediate step plays a more important 
role that one would have expected at first sight.\footnote{It is tempting to
 speculate that the process of localizing at the branch points is a B-model 
mirror analog to the process of using localization with respect to a torus action in Gromov-Witten theory.}

Even though it seems difficult to reformulate the recursion in a more
computationally effective way, what we can do is use our knowledge of the
modular properties of the amplitudes to rewrite the amplitudes \emph{a posteriori}. 
That is, using modularity and regularity of the amplitudes we write down a 
general ansatz for the amplitudes, either in terms of modular forms, or as functionals of solutions of the Picard-Fuchs equations. At each genus $g$ and number of holes $h$, the ansatz involves rational functions in the open and closed moduli comprising a finite number of unknown parameters. The latter can be fixed by comparing the ansatz with the result obtained from the recursion \eqref{e:rec1}. Alternatively, the parameters can be fixed by comparing with a mirror calculation at large radius using the topological vertex formalism \cite{topvertex}. 

These formulae 
prove to be very useful in studying the amplitudes at various points in 
the moduli space, as we will do in the next section. However, the rational
functions become rather involved and increasingly difficult to determine  for
higher genus and larger number of holes.

\subsubsection{Picard-Fuchs equations, monodromy and modularity of the closed amplitudes}

Given a family of Calabi-Yau threefolds $Y \to \CM$, it is standard to
associate a system of differential equations, called the \emph{Picard-Fuchs
  equations}, which annihilate periods of the holomorphic volume form
$\Omega_z$ on the fiber $Y_z$. In the noncompact setting, the Picard-Fuchs 
system can be extracted either by taking the limit of a compact
threefold~\cite{CKYZ}, or by considering the equivalent Landau-Ginzburg
setting~\cite{Gi,HV}. 
Solutions to the Picard-Fuchs equations provide a set of flat coordinates on $\CM$.

When $Y$ is of the form studied previously, it can be shown that the geometry
``reduces'' to the family of curves $\Sigma \to \CM$, and the Picard-Fuchs equations annihilate periods of the
one-form $\omega(p)$ over one-cycles on the Riemann surface $\Sigma_z$. 
From now on, we focus on the case where $\Sigma_z$ is a genus one curve. Let
$(A,B)$ be a canonical basis of one-cycles on the genus one curve $\Sigma_z$. 
Apart from a constant solution, there are two more solutions to the Picard-Fuchs equations, which provide a basis of dual periods:
\be
T = \oint_A \omega(p), \qquad T_D = \oint_B \omega(p).
\ee
The Picard-Fuchs differential equations have regular singular points, around which the periods have monodromy. The monodromy group is a finite index subgroup of $SL(2,\IZ)$.

A natural question is to study modularity of the B-model amplitudes with respect to the monodromy group. This question was approached for the closed amplitudes from physical principles in \cite{ABK}. The physical closed B-model amplitudes $\CF^{(g)}$ are invariant under the monodromy group --- indeed, this is required for consistency of the physical theory all over the moduli space $\CM$ --- but they are non-holomorphic. This can be reformulated in terms of modularity with respect to the modular parameter of $\Sigma_z$:
\be
\tau = \frac{\partial T_D}{\partial T} = \frac{\partial^2 \CF^{(0)}}{\partial T^2},
\ee
where $\CF^{(0)}$ is the prepotential of special geometry giving the genus $0$ closed B-model amplitude. In this language, the statement becomes that for $g \geq 2$, the physical amplitudes $\CF^{(g)}$ are \emph{almost holomorphic modular forms} with respect to the monodromy group \cite{ABK}. However, there is a canonical isomorphism between the ring of almost holomorphic modular forms and the ring of \emph{quasi-modular forms} --- forms that transform with a shift \cite{KZ}. This is given by ``taking the holomorphic limit'' of $\CF^{(g)}$, which breaks the modular invariance by keeping only the constant term in the finite expansion in $\Im(\tau)^{-1}$. We thus obtain the holomorphic closed B-model amplitudes $F^{(g)}$, which are quasi-modular with respect to the monodromy group. Thoses are the amplitudes that were constructed through the recursion.

\subsubsection{Modularity of the open amplitudes}

We now want to understand the modular properties of the open amplitudes $F^{(g,h)}$, which are the holomorphic limits of the monodromy invariant physical amplitudes $\CF^{(g,h)}$. This was studied in \cite{EMO,EO} using the recursion.

Let $\tau$ be the modular parameter of $\Sigma_z$, which parameterizes the upper half plane. Let
\be\label{e:mod}
\tilde \tau = \frac{A \tau + B}{C \tau +  D}, \qquad \begin{pmatrix} A & B \\ C & D \end{pmatrix} \in \Gamma \subset SL(2,\IZ)
\ee
be a symplectic transformation of the periods in the monodromy group $\Gamma$.

Under \eqref{e:mod}, the Bergman kernel transforms as
\be\label{e:kbk}
\tilde B(p,q)=B(p,q) -  2\pi \ri \, u(p) (C \tau + D)^{-1} C u(q),
\ee
where $u(p)$ is the holomorphic differential. The shift makes the Bergman kernel a quasi-modular form of weight $0$. Through the recursion \eqref{e:rec1}, this induces quasi-modular properties for all the open amplitudes $W^{(g,h)}$. One can compute the explicit transformation properties of the differentials $W^{(g,h)}$ by plugging in the transformation properties of the Bergman kernel directly in the recursion, as was done in \cite{EMO, EO}; we will not repeat the analysis here. Instead, what we are doing next is to use our knowledge of modularity to write down explicit expressions for the (low genus and number of hole) amplitudes in terms of modular forms, and as functionals of solutions of the Picard-Fuchs equations (the periods).

\subsection{Modular forms and functionals}

As we have just seen, the multilinear differentials $W^{(g,h)}$ are
quasi-modular forms of weight $0$ with respect to the monodromy group. 
They are obtained by taking the holomorphic limit of the non-holomorphic 
differentials $\CW^{(g,h)}$, which correspond to the physical amplitudes, 
therefore are monodromy invariant. As a consequence, the holomorphic
amplitudes $W^{(g,h)}$ can be universally written as functionals of the
periods and their derivatives, where the periods are functions of some local
coordinates on the moduli space. 

The functional point of view provides a very useful way of computing modular
transformations of the amplitudes, since changing the period in the functional
directly implements the symplectic transformation between the periods. 
In other words, the choice of period in the functional corresponds to a 
choice of modular parameter, or equivalently to a choice of canonical basis 
of cycles in the definition of the Bergman kernel. This approach 
renders the computation of the amplitudes everywhere in the moduli 
space straightforward.

To see how it goes, let us start by deriving a general expression for 
the annulus amplitude (or the Bergman kernel) in terms of modular forms, 
which is the main ingredient in the recursion relation, and induces the 
quasi-modular properties of the amplitudes. We then explain how it can be 
written generally as a functional; we will propose an exact form for the 
functional in the next section when we specialize to the local $\IP^2$
geometry. Finally we propose a general ansatz for the higher order amplitudes, 
which we use in the next section to derive functional expressions and compute 
the amplitudes at the orbifold point of local $\IP^2$.

\subsubsection{The annulus amplitude}

As usual, we start with a family of (punctured) Riemann surfaces $\Sigma \to \CM$ (the open B-model moduli space), and a choice of embedding of the fibers $\Sigma_z$ in $(\IC^*)^2$,
\be
\Sigma_z = \{ H(x,y;z)=0 \} \subset (\IC^*)^2.
\ee
We specialize to the case where $\Sigma_z$ is a genus one curve. Denote by $q_i \in \Sigma_z$ the ramification points of the projection map onto the $x$-axis, and by $\lambda_i := x(q_i) \in \IC^*$ the branch points.

When $\Sigma_z$ has genus one, the annulus amplitude $W^{(0,2)}$ can be written in terms of the Weierstrass elliptic function, using uniformization parameters for the elliptic curve. Alternatively, when $\Sigma_z$ has four distinct branch points $\lambda_i$, $i=1,\ldots,4$, one can work directly on the $\IC^*$ which is the image of the $x$-projection. In terms of $x$-projected variables $x_1,x_2 \in \IC^*$ (\emph{i.e.} local coordinates $x_1:=x(p_1)$ and $x_2:=x(p_2)$), Akemann derived in \cite{Ak} a formula for the annulus amplitude, which reads
\begin{align}
W^{(0,2)}(p_1,p_2) =& \frac{\rd x_1 \rd x_2}{4 \sqrt{\sigma(x_1)\sigma(x_2)}}\left(\frac{M(x_1,x_2)+M(x_2,x_1)}{(x_1-x_2)^2} - (\lambda_1 - \lambda_3)(\lambda_2-\lambda_4) \frac{E(k)}{K(k)} \right) \notag\\
&\quad- \frac{\rd x_1 \rd x_2}{2 (x_1-x_2)^2},
\label{e:ak}
\end{align}
where
\be
M(x_1,x_2)=(x_1-\lambda_1)(x_1-\lambda_2)(x_2-\lambda_3)(x_2-\lambda_4),
\ee
\be
\sigma(x) = \prod_{i=1}^4 (x-\lambda_i),
\ee
and $K(k)$ and $E(k)$ are elliptic integrals with modulus
\be
k^2 = \frac{(\lambda_1-\lambda_2)(\lambda_3-\lambda_4)}{(\lambda_1-\lambda_3)(\lambda_2-\lambda_4)}.
\ee
Note that the amplitude depends on a choice of ordering of the branch points, which corresponds to a choice of canonical basis of cycles on $\Sigma_z$.

Let us start by rewriting the amplitude in terms of modular forms.

\begin{prop}
Let
\be
S_k = \sum_{1 \leq j_1 < j_2 < \ldots < j_k \leq 4} \lambda_{j_1} \cdots \lambda_{j_k}
\ee
be the elementary symmetric polynomials in the four branch points, and let
\be\label{e:holdiff}
 u(x) \rd x = \ri \frac{\sqrt{(\lambda_1-\lambda_3)(\lambda_2-\lambda_4)}}{4 \sqrt{\sigma(x)} K(k)} \rd x
\ee
be the holomorphic differential.
The annulus amplitude can be written as
\be\label{e:modann}
W^{(0,2)}(p_1,p_2) = \left(- \frac{1}{2 (x_1-x_2)^2} + \frac{f_0^{(0,2)}(x_1,x_2)}{4 \sqrt{\sigma(x_1)\sigma(x_2)}} + \frac{\pi^2}{3} u(x_1) E_2(\tau) u(x_2) \right) \rd x_1 \rd x_2,
\ee
where $\tau$
is the modular parameter, $E_2(\tau)$ is the second Eisenstein series, and the rational function $f_0^{(0,2)}(x_1,x_2)$ reads\footnote{Recall from remark \ref{rem:dep} that we do not write explicitly the dependence on $z$ for simplicity.}
\be
f_0^{(0,2)}(x_1,x_2) = \frac{6 x_1^2 x_2^2 - 3 x_1 x_2(x_1+x_2) S_1 + (x_1^2 + 4 x_1 x_2 + x_2^2) S_2 - 3 (x_1+x_2) S_3 + 6 S_4}{3 (x_1-x_2)^2}.
\ee

\end{prop}

\begin{proof}
We start with Akemann's formula \eqref{e:ak}. Let us introduce
\be
\omega_1 = \frac{2 \ri}{\pi} \frac{K(k)}{\sqrt{(\lambda_1-\lambda_3)(\lambda_2-\lambda_4)}}, \qquad e_3 = \frac{1}{12}\left( S_2 - 3 (\lambda_1 \lambda_2 + \lambda_3 \lambda_4) \right).
\ee
$e_3$ is one of the three roots of the elliptic curve in Weierstrass form. Then, manipulating some of Akhiezer's identities for elliptic integrals \cite{Akh}, we obtain the identity
\be
E(k) K(k) = \pi^2 \left( \frac{1}{12} E_2 (\tau) + \omega_1^2 e_3 \right).
\ee
From this we rewrite the second term in \eqref{e:ak} as
\be
 - \frac{1}{4 \sqrt{\sigma(x_1)\sigma(x_2)}} (\lambda_1 - \lambda_3)(\lambda_2-\lambda_4) \frac{E(k)}{K(k)} = \frac{\pi^2}{3} u(x_1) E_2(\tau) u(x_2) + \frac{e_3}{\sqrt{\sigma(x_1)\sigma(x_2)}},
\ee
using the definition of the holomorphic differential above. By expanding the function $M(x_1,x_2)$ and combining with the $e_3$ term, we can rewrite the other terms of \eqref{e:ak} in terms of elementary symmetric polynomials of the branch points, and we obtain \eqref{e:modann}.
\end{proof}

\begin{rem}
In \eqref{e:modann}, the only term which is not quite modular invariant is the
term with $E_2(\tau)$. Since $E_2(\tau)$ is a quasi-modular form of weight
$2$, and the holomorphic differentials are modular of weight $-1$, we see
explicitly that the annulus amplitude is a quasi-modular form of weight $0$, 
as it should. The shift in the modular transformation of the annulus amplitude 
comes, as in the closed case \cite{ABK}, from the  shift in the modular 
transformation of the second Eisenstein series $E_2(\tau)$.
\end{rem}

\begin{rem}
Note that the function $f_0^{(0,2)}(x_1,x_2)$ is also rational in $z$ --- hence manifestly modular invariant --- since it involves only symmetric combinations of the branch points, which are necessarily rational functions of $z$. The function $f_0^{(0,2)}(x_1,x_2)$
corresponds to the ``holomorphic ambiguity'' in the integration of the holomorphic anomaly equation for the open amplitudes.
\end{rem}

Let us now define
\be
G(\tau) = \frac{E_2(\tau)}{3 \omega_1^2},
\label{e:calG}
\ee
which is a function of $z$ through the definition of $\omega_1$, and depends
on a choice of modular parameter $\tau$, but does not depend on the open
string variables  $x_1$ and $x_2$. The annulus can now be rewritten as
\be\label{e:annfunct}
W^{(0,2)}(p_1,p_2) = - \frac{\rd x_1 \rd x_2}{2 (x_1-x_2)^2} + 
\frac{f_0^{(0,2)}(x_1,x_2)+G(\tau)}{4 \sqrt{\sigma(x_1)\sigma(x_2)}}\rd x_1 \rd x_2.
\ee
The Bergman kernel $B(p_1,p_2)$ is obtained by changing the sign in front of the first term.

$G(\tau)$ plays an important role in the following, since it encodes the quasi-modular properties of the amplitudes. As a result, $G(\tau)$ can be expressed as a \emph{functional} of the period
$T$ and its derivative; in which case we will denote it as $G[T;z]$. The choice of period
$T$ corresponds to the choice of modular parameter $\tau$. The exact form of
$G[T;z]$ depends on the curve $\Sigma_z$; we will present it for the mirror of 
local $\IP^2$ in the next section.

To summarize, we now have an expression for the annulus amplitude in 
terms of modular forms, which can be rewritten as a functional of 
the period and its derivatives, using $G[T;z]$. Let us now study the 
higher order amplitudes.

\subsubsection{Higher amplitudes}

Now that we have a functional expression for the annulus amplitude
(\ref{e:annfunct}),  which is the main ingredient of the recursion, 
we can derive the principal functional form of the higher genus amplitudes
from (\ref{e:rec1}). 

\begin{lem}\label{l:genstruct}
For $g \geq 0$, $h \geq 1$, and $(g,h) \neq (0,1), (0,2)$, the general form 
of the amplitudes is
\be\label{general}
W^{(g,h)}(p_1,\ldots,p_h) = \frac{\rd x_1 \cdots \rd x_h}{\Delta^{2g-2+h} \prod_{i=1}^h \sqrt{\sigma(x_i)}} \sum_{i=0}^{3g-3+2h} G^i[T;z] f_i^{(g,h)}(x_1, \ldots,x_h),
\ee
where
\be
\Delta = \prod_{i<j}(\lambda_i-\lambda_j)^2
\ee
is the discriminant of the curve. The functions $f_i^{(g,h)}(x_1,
\ldots,x_h)$ are rational in their arguments and in the closed parameter $z$. 
Moreover, they have the form
\be
\label{e:holamb}
f_i^{(g,h)}(x_1, \ldots,x_h) = \frac{Q_i^{(g,h)}(x_1, \ldots,x_h)}{\left( \prod_{j=1}^h \sigma(x_j) \right)^{3g-2+h}} ,
\ee 
where the $Q_i^{(g,h)}(x_1, \ldots,x_h)$ are polynomials of finite degree in their arguments and in $z$.  
\end{lem} 

\begin{proof}[Sketch of the proof]
We obtain this general form by close inspection of the recursion \eqref{e:rec1}, and using the functional formula \eqref{e:annfunct} for the annulus amplitude. Let us simply sketch the main lines of the argument.

Let
\be
W^{(g,h)}(p_1,\ldots,p_h) = w^{(g,h)} (x_1, \ldots, x_h) \rd x_1 \cdots \rd x_h.
\ee
First, it is clear from the definition that the functions
\be
\sqrt{\sigma(x_1)} \cdots \sqrt{\sigma(x_h)} w^{(g,h)} (x_1, \ldots, x_h)
\ee
are rational in the $x_i$'s, since multiplying by the square roots amounts to cancelling the branch cuts.

Second, by pushing down the recursion \eqref{e:rec1} in order to obtain the analogs of Feynman rules, as in definition 4.5 of \cite{EO}, we see that each amplitude is represented by a graph with $3g-3+2h$ edges. Each edge gives a factor of either $B(p,q)$ or ${\dd E_{q, \bar q}(p) \over \omega(q) - \omega(\bar q)}$. Since both of these factors are polynomials of degree $1$ in $G[T;z]$, we obtain that $w^{(g,h)}$ must be a polynomial of order $3g-3+2h$ in $G[T;z]$. So what we know so far is that
\be
w^{(g,h)} (x_1, \ldots, x_h) = \frac{1}{\prod_{i=1}^h \sqrt{\sigma(x_i)}} \sum_{i=0}^{3g-3+2h} G^i[T;z] \tilde f_i^{(g,h)}(x_1, \ldots,x_h),
\ee
where the $\tilde f_i^{(g,h)}(x_1, \ldots,x_h)$ are rational in the $x_i$'s. It is also clear that the $\tilde f_i^{(g,h)}$ are rational in $z$, since we are summing over branch points, hence the $\tilde f_i^{(g,h)}$ can be expressed in terms of elementary symmetric polynomials in the branch points, which must be rational functions of $z$.

Finally, the denominators of the functions $\tilde f_i^{(g,h)}(x_1, \ldots,x_h)$ can be obtained from the pole structure of the integrand in the recursion \eqref{e:rec1}. The analysis is rather subtle, and we leave the details to the reader. Roughly speaking, after taking residues and summing over branch points, each pole of the form $\sigma(x)^{-k}$ contributes a factor of $\Delta^k$ in the denominator, and the double poles of the Bergman kernels combine to give the factors of $\sigma(x)$ in the denominator.
\end{proof}

For a particular geometry, by comparing the generic form of the amplitudes \eqref{general} with the explicit result obtained with the recursion, we can determine the functions $f_i^{(g, h)}(x_1, \ldots,x_h)$ at each genus and number of holes. Once this is done, the main advantage of the functional form of the amplitudes is that the computation of the amplitudes at various points in the moduli space  simply amounts  to inserting the right period $T$ in the functional. We exemplify this procedure in detail in the next section by studying the mirror of local $\IP^2$ at the $\IC^3 / \IZ_3$ orbifold point. 

Note that the general form of the amplitudes \eqref{general} was obtained directly by inspection of the recursion and using the functional formula for the annulus amplitude. Alternatively, it could have been obtained through direct integration of the open version of the holomorphic anomaly equation, in which case the functions $f_0^{(g,h)}(x_1, \ldots,x_h)$ would correspond to the holomorphic ambiguities. This complementary approach sheds new light on the structural constaints of the amplitudes coming directly from modularity; we hope to report on that in future work.

\section{Open orbifold Gromov-Witten invariants of $\IC^3 / \IZ_3$}

In this section we apply our formalism to the study of the mirror to local $\IP^2$ at the orbifold point in moduli space.

\subsection{Geometry}

We consider the geometry described in Example \ref{ex:localp2}. $\Sigma \to \CM$ is the one-parameter family of genus one Riemann surfaces with three punctures. We choose the following embedding for the fibers, 
\be\label{e:sigmaz}
\Sigma_z = \left \{y^2 + y(1+x) + z x^3 =0 \right \} \subset (\IC^*)^2.
\ee
We consider the B-model on this geometry, with B-branes wrapping the curve $C \subset Y$ as usual. The mirror theory is the A-model on the target space $X = K_{\IP^2}$, with a noncompact A-brane wrapping a special Lagragian submanifold of topology $\IR^2 \times S^1$ (see \cite{AV,AKV,BKMP} for a detailed description of these branes). The parameterization of the curve $\Sigma_z$ above corresponds on the A-model side to an ``outer brane with zero framing'', in the nomenclature of \cite{AKV,BKMP}. Unless specified, all our calculations in this section will be in this parameterization.

By mirror symmetry, the closed A-model moduli space is isomorphic to $\CM$. It has two patches, which correspond to two phases of the A-model. In each phase, there is a limit point near which the A-model amplitudes have a convergent expansion, and become the amplitudes of a non-linear sigma model (coupled to two-dimensional gravity). In the first patch, the limit point is the \emph{large radius point}, which is located at $z=0$. The amplitudes expanded near this point become generating functions of Gromov-Witten invariants of $X = K_{\IP^2}$. In the second patch, the limit point is the \emph{orbifold point}, located at $z= \infty$; a good local coordinate is $\psi = z^{-1/3}$. The amplitudes expanded near this point become generating functions of orbifold Gromov-Witten invariants of $X' = \IC^3 / \IZ_3$. As a result, moving from one patch to the other in $\CM$ induces a topologically-changing transition of the target space.

This analysis also extends to the open sector. In the large radius patch, the amplitudes $F^{(g,h)}$ expanded near the large radius point become generating functions for open Gromov-Witten invariants of $(X,L)$, where $L$ is the special Lagrangian submanifold mirror to $C$. Open Gromov-Witten invariants are defined in terms of stable maps from bordered Riemann surfaces with Lagrangian boundary conditions \cite{KL} --- see also \cite{GZ}. If $X$ admits a $U(1)$ action which fixes $L$, then the $U(1)$ acts naturally on the space of stable maps, and one can use localization to compute open Gromov-Witten invariants \cite{KL}. 

In the orbifold patch, one expects a similar story to hold, and the amplitudes expanded near the orbifold point should be generating functions for open orbifold Gromov-Witten invariants of $(X',L')$. Here, $L'$ is a Lagrangian submanifold of $\IC^3$ which is fixed by the $\IZ_3$ action, hence descends to a Lagrangian submanifold of the orbifold. This Lagrangian $L'$ corresponds to the original Lagrangian $L$ at the 
large radius point, and should exist as a consequence of the A-version of the McKay correspondence for derived categories. 
Therefore, one can consider stable maps from bordered Riemann surfaces to the orbifold $\IC^3 / \IZ_3$, in such a way that the boundaries are mapped to $L'$, and construct the corresponding open orbifold Gromov--Witten invariants. One could then follow the approach of \cite{KL} in the context of orbifolds, and use localization with respect to a $U(1)$ action to compute open orbifold Gromov-Witten invariants of $\IC^3/ \IZ_3$. Such open orbifold Gromov-Witten invariants have not been defined in the mathematical literature yet. However, Renzo Cavalieri informed us that he is presently working on this \cite{Ca}. In particular, he has managed to compute disk orbifold Gromov-Witten invariants of $\IC^3 / \IZ_3$ using localization of $\IZ_3$-Hodge integrals. His calculation matches perfectly with the results we present in subsection \ref{s:disk}, as we explain there.

Another useful point of view on these open orbifold Gromov--Witten 
invariants is to consider the topological string theory near the orbifold point as a perturbed ${\cal N}=2$ orbifold conformal field theory (CFT) coupled to gravity. 
From this point of view, 
the open topological string 
amplitude $F^{(g,h)}$ is a generating function of arbitrary 
insertions of bulk and boundary operators of the orbifold CFT. 
In the case of $\IC^3/\IZ_3$ there is only one bulk operator $\CO$. This is a twist operator which corresponds to a blow-up 
mode of the orbifold singularity, \emph{i.e.} to a 
deformation mode of the closed string modulus. In the presence of Lagrangian boundary conditions speficied by $L'$, one also has boundary preserving operators. These operators correspond to the 
insertion of open string states on the boundaries of the Riemann surface which maps to $L'$, and they are in one-to-one correspondence 
with $H^1(L, {\rm End}(E))$, where $E$ is an appropriate vector bundle on $L$ \cite{witten,aspinwall}. Since we have open strings with 
$h$ boundaries, the most general configuration can be obtained by considering 
$h$ branes wrapping $L'$. In our case $b_1(L')=1$, therefore there will be $h$ (integrated) boundary operators 
$\Psi_\ell$, $\ell=1, \cdots, h$, corresponding to the $h$ branes wrapping $L'$. We then 
have 
\be
F^{(g,h)} =\Big\langle \exp\Bigl( T_{orb} \CO + \sum_{\ell=1}^h X_\ell \Psi_\ell \Bigr) \Big\rangle_{g,h} =
\sum_{j, i_1, \cdots, i_h \ge 0}  \frac{1}{j!} { N^{(g,h)}_{(i_1, \cdots, i_h),j} T^j_{orb} X_{1}^{i_1}\cdots X_h^{i_h}}, 
\ee
where
\be
N^{(g,h)}_{(i_1, \cdots, i_h),j} ={1\over i_1! \cdots i_h!} \langle \CO^j  \Psi_1^{i_1} \cdots \Psi_h^{i_h} \rangle_{g,h}
\ee
and the vevs are calculated for the twisted ${\cal N}=2$ SCFT of the orbifold coupled to gravity on a Riemann surface $\Sigma_{g,h}$. The numbers $N^{(g,h)}_{(i_1, \cdots, i_h),j} $ should be identified with the open orbifold Gromov--Witten invariants. The combinatorial factor 
$i_1! \cdots i_h!$ is included in the invariant in order to agree with the conventions of Cavalieri for the open orbifold Gromov--Witten invariants 
to which we will compare our results later on. 

Our goal in this section is to use mirror symmetry and our B-model recursive formalism to compute generating functions of open orbifold Gromov-Witten invariants of $\IC^3 / \IZ_3$. This can be done in two ways; either by extracting the B-model amplitudes at the orbifold point from the large radius ones using the quasi-modular properties of the amplitudes, or by generating the amplitudes directly at the orbifold point using the functional expressions derived in the previous section. But before doing that, we need to understand the open and closed mirror maps near the orbifold point, in order to map the B-model amplitudes to the A-model amplitudes.

\subsection{Open and closed mirror maps}

\subsubsection{Closed mirror map}

The closed mirror map provides a local isomorphism between the closed A- and B-model moduli spaces. One needs to compute the flat coordinate near a given point of $\CM$, which is given by a solution of the associated Picard-Fuchs system. Inversion of the flat coordinate gives the closed mirror map. 

For the case under consideration, one obtains a single Picard-Fuchs equation, which reads
\be \label{e:PFclo}
\big(\Theta + 3 z (3 \Theta + 2 ) (3 \Theta+1) \big) \Theta f = 0,
\ee
with the logarithmic derivative $\Theta = z \frac{ \partial}{\partial z}$.

The constant function $f=1$ is always a solution of \eqref{e:PFclo}. Near $z=0$, the two other solutions are
\begin{align}\label{e:solns}
T(z) =& \log z -6 z + 45 z^2 - 560 z^3 + \ldots := \log z +\sigma(z), \notag \\
T_D(z) =& (\log z)^2 + 2 \sigma(z) \log z - 18 z + \frac{423}{2} z^2 - 2972 z^3 + \ldots,
\end{align}
$T(z)$ is the flat coordinate in the large radius patch. At $z=0$, $T(0) \to - \infty$, and the expansion parameter is set to $Q = \re^{T}$. The closed mirror map in this patch, which expresses $z$ in terms of the flat parameter $T$, is then given by
\be
z(Q) = Q + 6 Q^2 + 9 Q^3 + 56 Q^4 + \ldots
\ee

In the orbifold patch, the two non-trivial solutions to \eqref{e:PFclo} read
\be
B_k(\psi) = {(-1)^{k+1}  \psi^k \over k} ~_3 F_2 \left( {k \over 3}, {k \over 3}, {k \over 3}; {2 k \over 3}, 1 + {k \over 3}; \left( -{\psi \over 3} \right)^3 \right), 
\ee
with $k=1,2$, and we used the local coordinate $\psi = z^{-1/3}$. Using the explicit expansion of the hypergeometric system we get
\be
B_k(\psi)= \sum_{n \geq 0} {(-1)^{3 n+k+1} \psi^{3 n + k} \over (3 n + k)!} \left({ \Gamma \left( n + {k \over 3} \right) \over \Gamma \left( {k \over 3} \right) } \right)^3.
\ee
The flat parameter in this patch reads \cite{ABK}
\be
T_{orb}(\psi) = B_1(\psi),
\ee
and the dual period is $T_{orb,D}(\psi) = B_2(\psi)$. At $\psi=0$, we get $T_{orb}(0)=0$, hence $T_{orb}$ itself is a good expansion parameter. The closed mirror map reads
\be
\psi(T_{orb}) = T_{orb}+\frac{1}{648}T_{orb}^4-\frac{29 }{3674160}T_{orb}^7+\frac{6607
   }{71425670400}T_{orb}^{10} + \ldots
\ee

\subsubsection{The open mirror map}

The open mirror map extends the isomorphism to the open sector, which in the case under consideration is the fiber of the moduli space $\Sigma \to \CM$. Again, one needs to determine the open flat coordinate, which is a solution of the extended Picard-Fuchs system, as derived in \cite{LM,LMW}. The open mirror map is given by inverting this open flat coordinate. We refer the reader to \cite{BKMP,LM,LMW} for a detailed explanation of the extended Picard-Fuchs system.

In the large radius patch, it was shown in \cite{AKV,LM} that the open flat coordinate is given by
\be
X(x,z) = x \re^{\frac{1}{3}(\log z - T(z))},
\ee
where $x$ is the local coordinate $x$ on $\Sigma_z$, and $T(z)$ is the closed flat coordinate.
Note that $X(x,z)$ is monodromy-invariant under $z \mapsto \re^{2 \pi i} z$. At $(x,z)=(0,0)$, we have $X \to 0$, hence it is a good expansion parameter. The open mirror map becomes
\be
x(Q,X) = X (1 - 2 Q + 5 Q^2-32 Q^3 + \ldots).
\ee

In the orbifold patch, we argued in \cite{BKMP} --- by requiring that the disk amplitude, when expressed in flat coordinates, be monodromy-invariant under the $\IZ_3$ orbifold monodromy $\psi \mapsto \re^{2 \pi i/3} \psi$, which fixes the open flat coordinate uniquely, up to scale --- that the open flat coordinate must be given by
\be
X_{orb}(x,\psi) = x z^{1/3} = x \psi^{-1}.
\ee
The open mirror map simply becomes
\be\label{e:openorb}
x(X_{orb},T_{orb}) = X_{orb} \psi(T_{orb}),
\ee
where $\psi(T_{orb})$ is the closed mirror map.

\subsection{Quasi-modular transformations}

Let us start by computing the amplitudes explicitly, using the quasi-modular transformation of the amplitudes from large radius to the orbifold point.

\subsubsection{Disk amplitude}\label{s:disk}

For completeness, let us review the calculation of the orbifold disk amplitude, which was done in \cite{BKMP}. Recall that the disk amplitude is simply given by the Abel-Jacobi map
\be
F^{(0,1)} = \int \log(y(x)) \frac{\dd x}{x},
\ee
up to classical terms. $y(x)$ is obtained by solving the curve $\Sigma_z$ and keeping the relevant branch:
\be
y(x) = \frac{1}{2}\left(1 + x + {\sqrt{(1+x)^2  - 4\,z\,x^3}}\right).
\ee 
We want to expand the Abel-Jacobi map at the orbifold point. Remark that the open mirror map \eqref{e:openorb} is linear in $\psi$. Hence we must plug in the open mirror map before expanding in the closed coordinate to get a meaningful expansion. This being done, we get the orbifold disk amplitude
\be
F_{orb}^{(0,1)} = \sum_{i,j} \frac{1}{j!} N_{i,j}^{(0,1)} X_{orb}^i T_{orb}^j,
\ee
with the invariants given in table \ref{t:disk}.
\begin{table}
\begin{center}
\begin{footnotesize}
$\begin{array}{|c|cccccccccc|}
\hline
&&&&&i&&&&&\\
\hline
 j & 1 & 2 & 3 & 4 & 5 & 6 & 7 & 8 & 9 & 10 \\
\hline
 0 & 0 & 0 & -\frac{1}{3} & 0 & 0 & -\frac{1}{4} & 0 & 0 & -\frac{10}{27} & 0 \\
 1 & 1 & 0 & 0 & \frac{1}{2} & 0 & 0 & \frac{6}{7} & 0 & 0 & 2 \\
 2 & 0 & -\frac{1}{2} & 0 & 0 & -\frac{6}{5} & 0 & 0 & -\frac{15}{4} & 0 & 0 \\
 3 & 0 & 0 & \frac{2}{3} & 0 & 0 & 4 & 0 & 0 & 20 & 0 \\
 4 & \frac{1}{27} & 0 & 0 & -\frac{40}{27} & 0 & 0 & -\frac{154}{9} & 0 & 0 & -\frac{3400}{27} \\
 5 & 0 & -\frac{5}{54} & 0 & 0 & \frac{206}{45} & 0 & 0 & \frac{3215}{36} & 0 & 0 \\
 6 & 0 & 0 & \frac{10}{27} & 0 & 0 & -\frac{160}{9} & 0 & 0 & -\frac{4940}{9} & 0 \\
 7 & -\frac{29}{729} & 0 & 0 & -\frac{1432}{729} & 0 & 0 & \frac{19586}{243} & 0 & 0 & \frac{2820200}{729} \\
 8 & 0 & \frac{197}{1458} & 0 & 0 & \frac{15514}{1215} & 0 & 0 & -\frac{384575}{972} & 0 & 0 \\
 9 & 0 & 0 & -\frac{2}{3} & 0 & 0 & -\frac{292}{3} & 0 & 0 & \frac{5540}{3} & 0 \\
 10 & \frac{6607}{19683} & 0 & 0 & \frac{80456}{19683} & 0 & 0 & \frac{5544602}{6561} & 0 & 0 & -\frac{90503800}{19683} \\
 11 & 0 & -\frac{63107}{39366} & 0 & 0 & -\frac{945934}{32805} & 0 & 0 & -\frac{214690135}{26244} & 0 & 0 \\
 12 & 0 & 0 & \frac{8074}{729} & 0 & 0 & \frac{53768}{243} & 0 & 0 & \frac{21092500}{243} & 0 \\
 13 & -\frac{4736087}{531441} & 0 & 0 & -\frac{51705832}{531441} & 0 & 0 & -\frac{307254682}{177147} & 0 & 0 & -\frac{528718078600}{531441} \\
 14 & 0 & \frac{58248455}{1062882} & 0 & 0 & \frac{906117742}{885735} & 0 & 0 & \frac{8720423035}{708588} & 0 & 0\\
\hline
\end{array}$
\end{footnotesize}
\caption{Some invariants $N_{i, j}^{(0,1)}$ for the orbifold disk amplitude of $\IC^3 / \IZ_3$ at zero framing.}\label{t:disk}
\end{center}
\end{table}

\subsubsection{Incorporating framing}

As we mentioned earlier, the calculation above was done for an outer brane with zero framing. However, for the orbifold disk amplitude the calculation can be easily generalized to arbitrary framing. From a Gromov-Witten point of view, framing corresponds to a choice of torus action in the localization process. Hence the calculation at arbitrary framing is relevant for comparison with localization computations in Gromov-Witten theory.

Recall from \cite{BKMP} that a framing transformation of the brane is given by reparameterizing the embedding of the fibers $\Sigma_z$ in $(\IC^*)^2$ by
\be
(x_f, y_f) = (x y^f, y),
\ee
where $(x_f, y_f)$ are the new coordinates, and $f \in \IZ$ is the framing. In particular, the embedding of $\Sigma_z$ becomes
\be\label{e:sigmazfr}
\Sigma_z = \left \{ y_f^{3 f+2} + y_f^{3 f+1} + x_f y_f^{2 f+1} + z x_f^3 = 0 \right \} \subset (\IC^*)^2.
\ee
We compute the disk amplitude for this curve as
\be
F^{(0,1)}_f = \int \log(y_f(x_f) ) \frac{\rd x_f}{x_f},
\ee
where the function $y_f(x_f)$ is obtained by solving \eqref{e:sigmazfr} for $x_f$ (as a series expansion). Plugging in the mirror map, we obtain the invariants presented in table \ref{t:diskfr}, for general framing $f$.

\begin{table}
\begin{center}
\begin{footnotesize}
$\begin{array}{|c|cccc|}
\hline
&&&i&\\
\hline
 j & 1 & 2 & 3 & 4\\
\hline
 0 & 0 & 0 & -\frac{1}{3} & 0 \\
 1 & 1 & 0 & 0 & f+\frac{1}{2} \\
 2 & 0 & -f-\frac{1}{2} & 0 & 0 \\
 3 & 0 & 0 & 3 f^2+3 f+\frac{2}{3} & 0 \\
 4 & \frac{1}{27} & 0 & 0 & -\frac{8}{27} \left(54 f^3+81 f^2+37 f+5\right) \\
 5 & 0 & -\frac{5}{54} (2 f+1) & 0 & 0 \\
 6 & 0 & 0 & \frac{5}{27} \left(9 f^2+9 f+2\right) & 0 \\
 7 & -\frac{29}{729} & 0 & 0 & -\frac{8}{729} \left(1890 f^3+2835 f^2+1303 f+179\right) \\
 8 & 0 & \frac{197 (2 f+1)}{1458} & 0 & 0 \\
 9 & 0 & 0 & \frac{1}{3} \left(-9 f^2-9 f-2\right) & 0 \\
 10 & \frac{6607}{19683} & 0 & 0 & \frac{8 \left(102870 f^3+154305 f^2+71549 f+10057\right)}{19683} \\
 11 & 0 & -\frac{63107 (2 f+1)}{39366} & 0 & 0 \\
 12 & 0 & 0 & \frac{4037}{729} \left(9 f^2+9 f+2\right) & 0 \\
 13 & -\frac{4736087}{531441} & 0 & 0 & -\frac{8 \left(65783718 f^3+98675577 f^2+45818317 f+6463229\right)}{531441} \\
 14 & 0 & \frac{58248455 (2 f+1)}{1062882} & 0 & 0\\
\hline
\end{array}$
\end{footnotesize}
\caption{Some invariants $N_{i, j}^{(0,1)}$ for the orbifold disk amplitude of $\IC^3 / \IZ_3$ at general framing $f$.}\label{t:diskfr}
\end{center}
\end{table}

As we mentioned already, the framing $f$ is correlated to the choice of torus weights for localization of the Hodge integrals in Gromov--Witten theory. Renzo Cavalieri has implemented the Hodge integral calculation for the disk amplitude of $\IC^3 / \IZ_3$ \cite{Ca}. It turns out that the most natural choice of torus weights in Gromov--Witten theory does not correspond to $f=0$, but rather to $f=-2/3$ (or $f=-1/3$). To ease comparisons, we present in table \ref{t:disk23} the disk invariants for $f=-2/3$. Rather amazingly, these invariants are precisely equal to the invariants computed by Cavalieri in orbifold Gromov--Witten theory! Since the Gromov--Witten calculation is done on the A-model side, this comparison also shows that our choice of open orbifold mirror map \eqref{e:openorb}, which was argued in \cite{BKMP} from monodromy considerations, is correct, including the scale.

It may seem however odd to assign a non-integral value to $f$; it would be interesting to understand this issue better. Presumably, the denominator of $3$ comes from the orbifold $\IZ_3$ action at the orbifold point --- indeed, framing has so far only been interpreted from a large radius point of view in topological strings. Note however that non-integral framings have already been considered, although in a different context \cite{DF}.

\begin{table}
\begin{center}
$\begin{array}{|c|cccccccccc|}
\hline
&&&&&i&&&&&\\
\hline
 j & 1 & 2 & 3 & 4 & 5 & 6 & 7 & 8 & 9 & 10 \\
\hline
 0 & 0 & 0 & -\frac{1}{3} & 0 & 0 & \frac{1}{12} & 0 & 0 & -\frac{1}{27} & 0 \\
 1 & 1 & 0 & 0 & -\frac{1}{6} & 0 & 0 & \frac{5}{63} & 0 & 0 & -\frac{4}{81} \\
 2 & 0 & \frac{1}{6} & 0 & 0 & -\frac{4}{45} & 0 & 0 & \frac{7}{108} & 0 & 0 \\
 3 & 0 & 0 & 0 & 0 & 0 & 0 & 0 & 0 & 0 & 0 \\
 4 & \frac{1}{27} & 0 & 0 & -\frac{8}{81} & 0 & 0 & \frac{35}{243} & 0 & 0 & -\frac{400}{2187} \\
 5 & 0 & \frac{5}{162} & 0 & 0 & -\frac{188}{1215} & 0 & 0 & \frac{875}{2916} & 0 & 0 \\
 6 & 0 & 0 & 0 & 0 & 0 & 0 & 0 & 0 & 0 & 0 \\
 7 & -\frac{29}{729} & 0 & 0 & -\frac{248}{2187} & 0 & 0 & \frac{5705}{6561} & 0 & 0 & -\frac{146800}{59049} \\
 8 & 0 & -\frac{197}{4374} & 0 & 0 & -\frac{10972}{32805} & 0 & 0 & \frac{221221}{78732} & 0 & 0 \\
 9 & 0 & 0 & 0 & 0 & 0 & 0 & 0 & 0 & 0 & 0 \\
 10 & \frac{6607}{19683} & 0 & 0 & \frac{10984}{59049} & 0 & 0 & \frac{889805}{177147} & 0 & 0 & -\frac{74714800}{1594323} \\
 11 & 0 & \frac{63107}{118098} & 0 & 0 & \frac{385132}{885735} & 0 & 0 & \frac{51307949}{2125764} & 0 & 0 \\
 12 & 0 & 0 & 0 & 0 & 0 & 0 & 0 & 0 & 0 & 0 \\
 13 & -\frac{4736087}{531441} & 0 & 0 & -\frac{6768584}{1594323} & 0 & 0 & \frac{17027675}{4782969} & 0 & 0 & -\frac{33798787600}{43046721} \\
 14 & 0 & -\frac{58248455}{3188646} & 0 & 0 & -\frac{381155716}{23914845} & 0 & 0 & \frac{3576521095}{57395628} & 0 & 0\\
\hline
\end{array}$
\caption{Some invariants $N_{i, j}^{(0,1)}$ for the orbifold disk amplitude of $\IC^3 / \IZ_3$ at framing $f=-2/3$.}\label{t:disk23}
\end{center}
\end{table}

\subsubsection{Annulus amplitude}

We now want to compute the annulus amplitude, which is slightly more complicated, since it has non-trivial modular properties and transforms with a shift. More precisely, recall from \eqref{e:kbk} that the annulus transforms as
\be\label{e:annorb}
W_{orb}^{(0,2)}(p_1,p_2)=W^{(0,2)}(p_1,p_2) -  2\pi \ri \, u(p_1) (C \tau + D)^{-1} C u(p_2),
\ee
where $W^{(0,2)}(p_1,p_2)$ is the large radius annulus amplitude and $(C \tau + D)^{-1} C$ comes from the modular transformation of the period matrix $\tau$ from large radius to the orbifold.

The first step consists then in computing the annulus amplitude at large radius, using Akemann's formula \eqref{e:ak}. This was done in \cite{MM,BKMP}, and we will not repeat the calculation here. What we need to do however is to analytically continue this result to the orbifold point, to obtain the first term on the right hand side of \eqref{e:annorb}. The analytic continuation can be done directly in Akemann's formula, by expanding the branch points around $\psi=0$. However, it is important to note that as for the disk amplitude, we must write things in terms of the open flat coordinates $X_1$ and $X_2$ --- henceforth we will drop the subscript $orb$ --- before expanding in $\psi$, since the open mirror map is linear in $\psi$. 

After using a few identities involving elliptic functions and $\Gamma$-functions, we obtain the following analytic continuation of the large radius annulus amplitude to the orbifold point, in orbifold flat coordinates $X_1$ and $X_2$:
\begin{gather}
W^{(0,2)}=\rd X_1 \rd X_2 \left( \frac{-9\,{\sqrt{3}}\,{\Gamma(\frac{2}{3})}^6}{8\,{\pi }^3} - 
  \frac{81\,\psi\,{\Gamma(\frac{2}{3})}^{12}}{64\,{\pi }^6} + 
  {\psi}^2\,\left( \frac{1}{18} - \frac{243\,{\sqrt{3}}\,{\Gamma(\frac{2}{3})}^{18}}{512\,{\pi }^9} \right) \right. \notag\\  + 
  \left( 1 + \frac{9\,{\sqrt{3}}\,\psi\,{\Gamma(\frac{2}{3})}^6}{8\,{\pi }^3} + 
     \frac{81\,{\psi}^2\,{\Gamma(\frac{2}{3})}^{12}}{64\,{\pi }^6} \right) \,X_1 + 
  \left( -2\,\psi - \frac{9\,{\sqrt{3}}\,{\psi}^2\,{\Gamma(\frac{2}{3})}^6}{8\,{\pi }^3} \right) \,
   {X_1}^2 + \notag\\ \left( 1 + \frac{9\,{\sqrt{3}}\,\psi\,{\Gamma(\frac{2}{3})}^6}{8\,{\pi }^3} + 
     \frac{81\,{\psi}^2\,{\Gamma(\frac{2}{3})}^{12}}{64\,{\pi }^6} + 
     \left( -3\,\psi - \frac{9\,{\sqrt{3}}\,{\psi}^2\,{\Gamma(\frac{2}{3})}^6}{8\,{\pi }^3} \right) \,
      X_1 + 5\,{\psi}^2\,{X_1}^2 \right) \,X_2 \notag\\ \left. + 
  \left( -2\,\psi - \frac{9\,{\sqrt{3}}\,{\psi}^2\,{\Gamma(\frac{2}{3})}^6}{8\,{\pi }^3} + 
     5\,{\psi}^2\,X_1 + 3\,{X_1}^2 \right) \,{X_2}^2 +\ldots \right).
\label{e:anal}
\end{gather}
One can see that it is not rational, as expected; the non-rational terms should be cancelled by the shift in \eqref{e:annorb}.

The next step is to compute the modular transformation between the large radius periods $(T,T_D)$ and the orbifold periods $(T_{orb}, T_{orb,D})$. This can be done by standard analytic continuation, as in \cite{ABK}. Define
\be
c_1 = -{1 \over 2 \pi i} {\Gamma(1/3) \over \Gamma(2/3)^2},\qquad c_2 = {1 \over 2 \pi i} {\Gamma(2/3) \over \Gamma(1/3)^2}, \qquad \beta = {1 \over (2 \pi i)^3}, \qquad \omega = \re^{2 \pi i/3}.
\ee
We get the transformation
\be
\begin{pmatrix} T_D \\ T \\ 1 \end{pmatrix} = \begin{pmatrix} {\beta \omega^2 \over c_1} & {\beta \omega \over c_2} & {1 \over 3} \\ - c_2 & c_1 & 0 \\ 0 & 0 & 1 \end{pmatrix} \begin{pmatrix} T_{orb,D} \\ T_{orb} \\ 1 \end{pmatrix}.
\label{symptrans}
\ee
Note that this transformation is not quite symplectic, since its determinant is $- \beta$; that is, it changes the scale of the symplectic form. However, this can be taken into account by renormalizing the string coupling constant, as in \cite{ABK}.

Now the modular transformation that we want is given by the inverse of this matrix. We get that
\be
C = -\frac{c_2}{\beta}, \qquad D = - \frac{\omega^2}{c_1}.
\ee
We also need the large radius period matrix $\tau$, analytically continued around $\psi=0$. By definition, it is given by
\be
\tau (\psi)= \frac{\partial T_D}{\partial T} = \frac{ \partial T_D / \partial \psi}{\partial T / \partial \psi}.
\ee
Using the transformation above between $T_D, T$ and $T_{orb,D}, T_{orb}$, and expanding around $\psi=0$, we get
\be
\tau (\psi) = - \frac{{\left( -1 \right) }^{\frac{1}{6}}}{{\sqrt{3}}}  - 
  \frac{i \,2^{\frac{1}{3}}\,\psi\,{\Gamma(\frac{2}{3})}^2}{{\Gamma(\frac{1}{6})}^2} - 
  \frac{i \,{\psi}^2\,{\Gamma(\frac{2}{3})}^7}{2 \pi \,{\Gamma(\frac{1}{3})}^5} + \CO(\psi^3).
\ee
Finally, we can compute the holomorphic differential $u(p)$ from the standard formula \eqref{e:holdiff}.

Putting all this together, and integrating, we obtain the orbifold annulus amplitude in flat orbifold coordinates $X_1,X_2,T_{orb}$
\be\label{e:annres}
F_{orb}^{(0,2)} = \sum_{i_1,i_2,j} \frac{1}{j!} N_{(i_1,i_2),j}^{(0,2)} X_1^{i_1} X_2^{i_2} T_{orb}^j,
\ee
with the invariants $N_{(i_1,i_2),j}^{(0,2)}$ given in table \ref{t:ann}; the invariants are symmetric in $(i_1,i_2)$.

\begin{table}
\begin{center}
$\begin{array}{|c|ccccccccc|}
\hline
&&&&&(i_1,i_2)&&&&\\
\hline
 j & \text{(1,1)} & \text{(2,1)} & \text{(3,1)} & \text{(2,2)} & \text{(4,1)} & \text{(3,2)} & \text{(5,1)} & \text{(4,2)} & \text{(3,3)} \\
 0 & 0 & \frac{1}{2} & 0 & 0 & 0 & 0 & \frac{3}{5} & \frac{1}{2} & \frac{1}{3} \\
 1 & 0 & 0 & -\frac{2}{3} & -\frac{3}{4} & 0 & 0 & 0 & 0 & 0 \\
 2 & \frac{1}{9} & 0 & 0 & 0 & \frac{14}{9} & \frac{5}{3} & 0 & 0 & 0 \\
 3 & 0 & -\frac{1}{6} & 0 & 0 & 0 & 0 & -\frac{26}{5} & -\frac{16}{3} & -\frac{16}{3} \\
 4 & 0 & 0 & \frac{34}{81} & \frac{11}{36} & 0 & 0 & 0 & 0 & 0 \\
 5 & -\frac{1}{243} & 0 & 0 & 0 & -\frac{338}{243} & -\frac{65}{81} & 0 & 0 & 0 \\
 6 & 0 & -\frac{1}{54} & 0 & 0 & 0 & 0 & \frac{238}{45} & \frac{56}{27} & \frac{40}{27} \\
 7 & 0 & 0 & \frac{562}{2187} & \frac{197}{972} & 0 & 0 & 0 & 0 & 0 \\
 8 & \frac{391}{6561} & 0 & 0 & 0 & -\frac{17206}{6561} & -\frac{4261}{2187} & 0 & 0 & 0 \\
 9 & 0 & -\frac{29}{162} & 0 & 0 & 0 & 0 & \frac{3614}{135} & \frac{1552}{81} & \frac{160}{9} \\
 10 & 0 & 0 & \frac{31606}{59049} & \frac{8333}{26244} & 0 & 0 & 0 & 0 & 0 \\
 11 & -\frac{225595}{177147} & 0 & 0 & 0 & \frac{30802}{177147} & \frac{158125}{59049} & 0 & 0 & 0 \\
 12 & 0 & \frac{8455}{1458} & 0 & 0 & 0 & 0 & -\frac{44338}{1215} & -\frac{48104}{729} & -\frac{52712}{729} \\
 13 & 0 & 0 & -\frac{49954466}{1594323} & -\frac{15072793}{708588} & 0 & 0 & 0 & 0 & 0 \\
 14 & \frac{301065409}{4782969} & 0 & 0 & 0 & \frac{712334462}{4782969} & \frac{8347925}{1594323} & 0 & 0 & 0\\
\hline
\end{array}$
\caption{Some invariants $N_{(i_1,i_2), j}^{(0,2)}$ for the orbifold annulus amplitude of $\IC^3 / \IZ_3$.}\label{t:ann}
\end{center}
\end{table}

The invariants are rational, as they should. Moreover, it is easy to see that the amplitude is invariant under the $\IZ_3$ orbifold monodromy. Indeed, the orbifold monodromy is given by
\be
(T_{orb},X_1,X_2) \mapsto (\omega T_{orb},\omega^2 X_1, \omega^2 X_2), \qquad \omega = \re^{2 \pi i/3}.
\ee
Thus all terms in the expansion above are monodromy invariant. 

In table \ref{t:fann} we also present some results for the corresponding framed invariants. 

\begin{table}
\begin{center}
\begin{footnotesize}
$\begin{array}{|c|cccc|}
\hline
&&&(i_1,i_2)&\\
\hline
 j & \text{(1,1)} & \text{(2,1)} & \text{(3,1)} & \text{(2,2)} \\
 0 & 0 & \frac{1}{2} +f& 0 & 0  \\
 1 & 0 & 0 & -\frac{2}{3} +f(2f-1)& -\frac{3}{4}+f(2f-1) \\
 2 & \frac{1}{9} +f(f+1) & 0 & 0 & 0  \\
 3 & 0 &
 -\frac{1}{6}-\frac{1}{3} f \left(12 f^2+18 f+7\right)& 0 & 0  \\
 4 & 0 & 0 & \frac{34}{81}- {f( 2646 f^3+2592 f^2+538 f-53) \over 27}& \frac{11}{36} - {f(2727 f^3+2754 f^2+637 f-35)\over 27} \\
 5 & -\frac{1}{243} +{5 f(1+f)\over 27}& 0 & 0 & 0  \\
 6 & 0 & -\frac{1}{54} -{f(31+90 f+ 60 f^2) \over 27}& 0 & 0 \\
\hline
\end{array}$
\end{footnotesize}
\caption{Some invariants $N_{(i_1,i_2), j}^{(0,2)}$ for the framed orbifold annulus amplitude of $\IC^3 / \IZ_3$.}\label{t:fann}
\end{center}
\end{table}

\subsubsection{Higher amplitudes}

Computing the higher amplitudes directly using the modular shift is rather complicated, partially because of all the elliptic functions involved in the calculation. It is much simpler to use the functional expressions to compute the orbifold amplitudes. We have however checked that the genus 0, three-hole amplitude computed through the shift also matches the functional calculation, but we will not present the calculation here for brevity.

\subsection{Calculation using the functionals}

Let us now use the functional expressions for the amplitude derived in the previous section to compute the open orbifold amplitudes. First, we need to specify what the functional $G[T;z]$ is for the curve $\Sigma_z$ given by \eqref{e:sigmaz}.

\subsubsection{Generalities}

First, from the embedding of the elliptic curve \eqref{e:sigmaz}, we obtain
\be
\sigma(x)= (x + 1)^2 - 4 x^3 z,
\ee
and the discriminant
\be\label{e:disc}
\Delta = 1 + 27 z.
\ee
We claim that the functional $G[T;z]$ reads
\be\label{e:functG}
G[T;z] = - \frac{1}{z^2 C_{z z z} } \frac{\partial}{\partial z} \left(4 \log \frac{\partial T}{\partial z} +  \frac{1}{3} \log \Delta + 5 \log z \right),
\ee
where 
\be
C_{z z z} = \frac{\partial^3 F^{(0)}}{\partial z^3} = \frac{3}{z^3 \Delta}
\ee
is the Yukawa coupling in the local variable $z$. Let us sketch the derivation of this functional formula.

The genus one amplitude $F^{(1)}$ was defined in definition \ref{def:f1}. In our context, one can show that \eqref{e:deff1} becomes, up to a constant term,
\be
F^{(1)} = -\log \eta(\tau)-\frac{1}{24}\log \tilde \Delta(\psi),
\ee
where $\tilde \Delta(\psi) = 27 + \psi^3$ is the discriminant in terms of $\psi = z^{-1/3}$, and $\eta(\tau)$ is the Dedekind $\eta$-function. Alternatively, $F^{(1)}$ can also be expressed as \cite{BCOV2}
\be
F^{(1)} = -\frac{1}{2}\log\frac{\partial T}{\partial  \psi} 
-\frac{1}{12}\log \tilde \Delta(\psi).
\ee
Combining the two formulae, we obtain
\be
\log \eta(\tau)  =\frac{1}{2}\log \frac{\partial T}{\partial  \psi} +\frac{1}{24}\log \tilde \Delta(\psi) .
\ee
Now the second Eisenstein series $E_2(\tau)$ is related to the Dedekind $\eta$-function by:
\be
E_2(\tau) =24 \frac{d}{d\tau} \log \eta(\tau).
\ee
As a result, we get
\be
 E_2(\tau)=\frac{\partial}{\partial \tau} \left(12 \log \frac{\partial T}{\partial \psi} + \log \tilde\Delta(\psi) \right).
\ee
Using the fact that
\be
\tau = \frac{\partial^2 F^{(0)}}{\partial T^2}, \qquad \tilde \Delta(z^{-1/3}) = \frac{\Delta}{z},
\ee
we obtain
\be
E_2(\tau) = \left( \frac{\partial T}{\partial z} \right)^2 \frac{1}{C_{z z z}} \frac{\partial}{\partial z} \left(12 \log \frac{\partial T}{\partial z} + \log \Delta + 15 \log z \right).
\ee
Finally, recall that $G[T;z]$ is defined by
\be
G[T;z] = \frac{E_2(\tau)}{3 \omega_1^2}.
\ee
By direct computation, we can write $\omega_1$ as a functional of $T$ and $z$,
\be
\omega_1 = \ri z \frac{\partial T}{\partial z},
\ee 
and we obtain the final formula for $G[T;z]$ given in \eqref{e:functG}.

With this explicit formula for the functional $G[T;z]$, we can proceed with the calculation of the higher amplitudes, using our ansatz \eqref{general}. As explained previously, to compute the amplitudes at the orbifold point, all that we need to do is to input the period $T_{orb}$ corresponding to the flat parameter at the orbifold point in the functional.

\subsubsection{Annulus amplitude}

The functional expression for the annulus amplitude was obtained in \eqref{e:annfunct}, using the expression \eqref{e:functG} for $G[T;z]$. For the curve $\Sigma_z$ under consideration, the rational function $f_0^{(0,2)}(x_1,x_2)$ can be computed, and reads
\be
f_0^{(0,2)}(x_1, x_2) = \frac{6 + 6 x_2 + x_2^2 + x_1^2 (1 - 12 x_2 z) + x_1 (6 + 4 x_2 - 12 x_2^2 z)}{3 (x_1 - x_2)^2}.
\ee

All that one needs to do to obtain the orbifold annulus amplitude, is to do the change of variable $z = \psi^{-3}$, replace the open moduli $x_1$ and $x_2$ by the open orbifold mirror map $x_{1,2}= X_{1,2} \psi$, where $X_1$ and $X_2$ are the open flat coordinates, and insert the closed flat orbifold coordinate $T= T_{orb}(\psi)$ in the functional. Then, we plug in the closed mirror map in the result and expand in $X_1$, $X_2$ and $T_{orb}$ to obtain the orbifold annulus amplitude. It is easy to show that we obtain precisely \eqref{e:annres} with the invariants of table \ref{t:ann}; note however how much simpler the calculation was.

\subsubsection{Genus $1$, one-hole}

The amplitude has the form predicted by the ansatz \eqref{general}. By comparing with the result obtained through the recursion, we can fix the functions $f_i^{(1,1)}(x)$. We obtain
\be\label{e:11p2}
W^{(1,1)} = \frac{\rd x}{\sqrt{\sigma(x)} \Delta} \left(\frac{9}{32 } G^2[T;z] + f_1^{(1,1)}(x) G[T;z] + f_0^{(1,1)}(x)\right),
\ee
with the functions:
\begin{align}
f_1^{(1,1)}(x)=&\frac{x(1+x)\Delta}{8 \sigma(x)}, \notag\\
f_0^{(1,1)}(x)=&\frac{1}{96 \sigma(x)^2}(1 + 36 z + 4 x (1 + 36 z) + 16 x^6 z^2 (1 + 36 z) + 
  6 x^2 (1 + 46 z + 270 z^2) \notag\\
&+ x^4 (1 + 56 z + 396 z^2) + 
  4 x^3 (1 + 55 z + 495 z^2) + 4 x^5 z (1 + 57 z + 1296 z^2)).
\end{align}

Doing the transformations as above to go to the orbifold point, we obtain the amplitude
\be
F_{orb}^{(1,1)} = \sum_{i,j} \frac{1}{j!} N_{i,j}^{(1,1)} X^i T_{orb}^j,
\ee
with the invariants given in table \ref{t:11}.
\begin{table}
\begin{center}
\begin{footnotesize}
$\begin{array}{|c|ccccccccc|}
\hline
&&&&&i&&&&\\
\hline
 j & 1 & 2 & 3 & 4 & 5 & 6 &7&8&9\\
\hline
 0 & 0 & 0 & \frac{5}{24} & 0 & 0 & \frac{11}{8} & 0 & 0 & \frac{85}{12} \\
 1 & \frac{1}{72} & 0 & 0 & -\frac{5}{9} & 0 & 0 & -\frac{77}{12} & 0 & 0 \\
 2 & 0 & -\frac{1}{36} & 0 & 0 & \frac{25}{12} & 0 & 0 & \frac{110}{3} & 0 \\
 3 & 0 & 0 & \frac{1}{12} & 0 & 0 & -10 & 0 & 0 & -\frac{495}{2} \\
 4 & \frac{1}{1944} & 0 & 0 & -\frac{86}{243} & 0 & 0 & \frac{18823}{324} & 0 & 0 \\
 5 & 0 & -\frac{11}{972} & 0 & 0 & \frac{3301}{1620} & 0 & 0 & -\frac{127415}{324} & 0 \\
 6 & 0 & 0 & \frac{31}{324} & 0 & 0 & -\frac{412}{27} & 0 & 0 & \frac{162755}{54} \\
 7 & \frac{475}{52488} & 0 & 0 & -\frac{5210}{6561} & 0 & 0 & \frac{1237285}{8748} & 0 & 0 \\
 8 & 0 & -\frac{223}{26244} & 0 & 0 & \frac{307847}{43740} & 0 & 0 & -\frac{6757145}{4374} & 0 \\
 9 & 0 & 0 & -\frac{1}{12} & 0 & 0 & -\frac{610}{9} & 0 & 0 & \frac{344095}{18} \\
 10 & -\frac{395585}{1417176} & 0 & 0 & \frac{172678}{177147} & 0 & 0 & \frac{168774025}{236196} & 0 & 0 \\
 11 & 0 & \frac{712639}{708588} & 0 & 0 & -\frac{5242661}{1180980} & 0 & 0 & -\frac{1966276115}{236196} & 0 \\
 12 & 0 & 0 & -\frac{38945}{8748} & 0 & 0 & -\frac{62488}{729} & 0 & 0 & \frac{158337275}{1458} \\
 13 & \frac{640118305}{38263752} & 0 & 0 & \frac{133378114}{4782969} & 0 & 0 & \frac{23152439695}{6377292} & 0 & 0 \\
 14 & 0 & -\frac{1726238977}{19131876} & 0 & 0 & -\frac{11317800859}{31886460} & 0 & 0 & -\frac{152933889775}{1594323} & 0\\
\hline
\end{array}$
\end{footnotesize}
\caption{Some invariants $N_{i, j}^{(1,1)}$ for the genus $1$, $1$ hole orbifold amplitude of $\IC^3 / \IZ_3$.}\label{t:11}
\end{center}
\end{table}

\subsubsection{Genus $0$, three-hole}

The amplitude has again the form predicted by the ansatz \eqref{general}. We can fix the functions $f_i^{(0,3)}(x_1,x_2,x_3)$ by comparing with the recursion, and we obtain
\begin{align}\label{e:30exp}
W^{(0,3)} = \frac{\rd x_1 \rd x_2 \rd x_3}{\sqrt{\sigma(x_1)\sigma(x_2)\sigma(x_3)} \Delta} \Big(&\frac{9}{64} G^3[T;z] + f_2^{(0,3)}(x_1,x_2,x_3) G^2[T;z]\notag\\ &+ f_1^{(0,3)}(x_1,x_2,x_3) G[T;z] +f_0^{(0,3)}(x_1,x_2,x_3) \Big).
\end{align}
The functions $f_i^{(0,3)}(x_1,x_2,x_3)$ are rather complicated; we present them in Appendix B.

Doing the transformations as above to go to the orbifold point, we obtain the amplitude
\be
F_{orb}^{(0,3)} = \sum_{i_1,i_2,i_3,j} \frac{1}{j!} N_{(i_1,i_2,i_3),j}^{(0,3)} X_1^{i_1} X_2^{i_2} X_3^{i_3} T_{orb}^j,
\ee
with the invariants given in table \ref{t:03}. The invariants are symmetric in $(i_1,i_2,i_3)$.

\begin{table}
\begin{center}
$\begin{array}{|c|ccccccc|}
\hline
 &  & && (i_1,i_2,i_3)  &&&\\
\hline
j & (1,1,1) & (2,1,1) & (3,1,1) & (2,2,1) &(4,1,1) & (3,2,1) & (2,2,2) \\
\hline
 0 & \frac{2}{3} & 0 & 0 & 0 & \frac{4}{3} & 1 & \frac{9}{8} \\
 1 & 0 & -\frac{5}{6} & 0 & 0 & 0 & 0 & 0 \\
 2 & 0 & 0 & \frac{52}{27} & \frac{23}{12} & 0 & 0 & 0 \\
 3 & -\frac{1}{27} & 0 & 0 & 0 & -\frac{176}{27} & -\frac{19}{3} & -\frac{49}{8} \\
 4 & 0 & \frac{13}{162} & 0 & 0 & 0 & 0 & 0 \\
 5 & 0 & 0 & -\frac{124}{729} & -\frac{11}{324} & 0 & 0 & 0 \\
 6 & -\frac{1}{81} & 0 & 0 & 0 & -\frac{32}{81} & -\frac{37}{27} & -\frac{133}{72} \\
 7 & 0 & \frac{397}{4374} & 0 & 0 & 0 & 0 & 0 \\
 8 & 0 & 0 & -\frac{17972}{19683} & -\frac{7303}{8748} & 0 & 0 & 0 \\
 9 & \frac{37}{243} & 0 & 0 & 0 & \frac{2480}{243} & \frac{709}{81} & \frac{1771}{216} \\
 10 & 0 & -\frac{92273}{118098} & 0 & 0 & 0 & 0 & 0 \\
 11 & 0 & 0 & \frac{3393164}{531441} & \frac{1584895}{236196} & 0 & 0 & 0 \\
 12 & -\frac{42703}{6561} & 0 & 0 & 0 & -\frac{475424}{6561} & -\frac{57857}{729} & -\frac{172613}{1944} \\
 13 & 0 & \frac{120276571}{3188646} & 0 & 0 & 0 & 0 & 0 \\
 14 & 0 & 0 & -\frac{4470350924}{14348907} & -\frac{1939962841}{6377292} & 0 & 0 & 0\\
\hline
\end{array}$
\caption{Some invariants $N_{(i_1,i_2,i_3), j}^{(0,3)}$ for the genus $0$, $3$ hole orbifold amplitude of $\IC^3 / \IZ_3$.}\label{t:03}
\end{center}
\end{table}

\subsection{Conifold point}\label{s:con}

So far we considered the A- and B-model amplitudes in the two distinct phases of $\CM$, namely the large radius phase and the orbifold  one.
There is however a third point around which the amplitudes have an interesting expansion, which is the \emph{conifold point}. This is not a limit point of a phase of $\CM$; rather, it is a singular point of the moduli space, where the target space of the A-model develops a conifold singularity.\footnote{In the gauged linear sigma model description of the A-model, at the conifold point new massless modes appear, which defines a new branch of vacua.}
This point is located at $z = - \frac{1}{27}$.

It is generally interesting to expand the amplitudes near the conifold point. For instance, the leading behavior of the closed amplitudes $F^{(g)}_{con}$ expanded at the conifold point can be understood as the amplitudes of non-critical $c=1$ string at the self-dual radius \cite{GVc1}. Moreover, the amplitudes $F^{(g)}_{con}$ seem to possess a universal gap, as discovered in \cite{HK}.
That is, the leading behavior of the closed amplitudes is of the form:
\bea 
F_{con}^{(g)}={B_{2g}\over 2 g ( 2g -2) T_{con}^{2g-2}}+ k^{(g)}_1  T_{con}+
{\cal O}(T_{con}^2)\ ,
\label{gap}
\eea
where the $B_{n}$ are the Bernoulli numbers and $T_{con}$
is the vanishing period at the conifold. 
This feature is rather striking, and very useful computationally. Indeed, one of the most effective approach for computing closed amplitudes is by directly integrating \cite{GKMW}  the holomorphic anomaly equation of \cite{BCOV}, using the polynomial structure of the amplitudes proposed in \cite{YY}. However, the holomorphic anomaly equation is not complete; at each genus one needs to fix a finite number of constants (the holomorphic ambiguity) using extra data. In conjunction with the leading behavior of the amplitudes, the gap behavior at the conifold point --- more precisely the absence 
of the $2g-3$ subleading negative powers 
in the $T_{con}$ expansion --- imposes
$2g-2$ such extra conditions, which have been shown to completely fix the holomorphic ambiguity in many local geometries. In the compact setting, they allow computation of closed amplitudes to very high genus \cite{HKQ}.

One may wonder if this approach  has an open counterpart.
So far, we relied entirely  on the recursion formalism to 
compute open amplitudes. As we have seen, while this formalism is very satisfactory conceptually, it is rather cumbersome computationally. Direct integration of the
open  holomorphic anomaly equations --- recently derived in \cite{EMO} in the local setting --- would provide an alternative method to compute the open amplitudes.
In particular, one could hope that a gap behavior exists for the open amplitudes expanded at the conifold point, providing sufficient boundary conditions  to fix the holomorphic ambiguity. 

This is surely enough motivation for studying in more detail the open amplitudes near the conifold point. In what follows we present general properties of the amplitudes; technical and computational aspects are relegated to Appendix C.

Consider as usual the moduli space $\Sigma \to \CM$. As mentioned before, the conifold point in the closed moduli space $\CM$ is located at $z=-1/27$; a good local coordinate is
\be
w = 27 z + 1.
\ee
Since we are computing open amplitudes, we must also specify where we expand the amplitudes in the open moduli space $\Sigma_w$.
At the conifold point $w=0$, it turns out that two of the branch points of the $x$-projection of the curve $\Sigma_w$ collapse to the same value, $x=-1/3$. Instead of expanding the open amplitudes near $x=0$, we will now expand the amplitudes near this critical point $x=-1/3$, using a new local coordinate centered at this point:\footnote{Note however that this critical point is a singular limit of the curve $\Sigma_w$, hence one has to choose an appropriate set of coordinates to smooth out the singularity. In particular, one must consider a double--scaling limit,
where the open coordinate $p$ is rescaled with the closed modulus, as explained in \cite{EO}. We will come back to that in the explicit computations in Appendix C.\label{f:singular}}
\be
p = \frac{1}{x}+ 3.
\ee
On the mirror A-model side, expanding the amplitudes near this point should correspond to considering branes located near a vertex of the toric diagram.

The open B-model amplitudes expanded near this critical point should correspond, at leading order, to $c=1$ string amplitudes at the self-dual radius,\footnote{Note that such critical points have already been considered in the context of matrix models. In \cite{dgkv} it was proposed  that $c=1$ amplitudes can be obtained in a two-cut matrix model by considering the critical limit where the two cuts touch each other.}
which are in turn   equivalent  to Gaussian matrix model amplitudes \cite{dv02}.
More precisely, the expected leading behavior of the open amplitudes near the critical point is \cite{bh}:
\be\label{e:leading}
F^{(g,h)} \sim T_{con}^{2-2g-h} \tilde F^{(g,h)},
\ee
where $T_{con}$ is the closed flat coordinate near the conifold point $w=0$, which corresponds to the vanishing period. The amplitudes $\tilde F^{(g,h)}$, which are independent of $T_{con}$, are to be identified with the amplitudes for FZZT branes in the $c=1$ string at the self-dual radius. 
Indeed, as noticed  in a similar context in \cite{MM}, in the critical limit
 toric branes should become FZZT branes.

The leading behavior (\ref{e:leading}) is the open analog to the leading behavior for the closed amplitudes, as presented in \eqref{gap}. Recalling the discussion above for the closed amplitudes, the open amplitudes would possess a gap if the subleading terms in $T_{con}$ with negative exponents vanished.

We can use the recursion and the formalism developed in section 2 to compute the B-model amplitudes explicitly at the critical point; we report this calculation in Appendix C. The calculation shows that the amplitudes indeed possess the expected  leading behavior in $T_{con}$. However, the subleading terms in $T_{con}$ are not vanishing, in contrast with the closed amplitudes. As a result, we conclude that in the open case, there is no simple gap behavior at the conifold. This renders the use of the direct integration of the holomorphic anomaly equations as a method to solve for the amplitudes rather limited in the open case, since one lacks the boundary conditions provided by the gap behavior and required to fix the holomorphic ambiguity.

A word of caution to end this section; as we discuss in Appendix C, it is not clear to us how to fix the open flat coordinate near the critical point $(w,p)=(0,0)$. This prevents us from providing unambiguous results for the open amplitudes near the conifold point.
 It would be interesting to clarify these issues further.

\appendix

\section{Landau-Ginzburg vs sigma model}

In this Appendix we explain the relation between the standard Landau-Ginzburg mirrors to toric threefolds and the sigma models described previously. We follow the argument presented by Hori, Iqbal and Vafa in p.93 of \cite{HIV}. 

Consider the A-twisted sigma model on a (noncompact) toric Calabi-Yau threefold $X$ defined by the toric charge vectors $Q^a$, $a=1,\ldots,k$. Its mirror \cite{Gi,HV} --- see also \cite{CCIT} for a clear explanation --- is a B-twisted Landau-Ginzburg model on the family of algebraic tori $\pi: V \to \CM$, with $V = (\IC^*)^{3+k}$, and $\CM = (\IC^*)^k$, with projection map
\be
\pi: (y_0, \ldots, y_{k+2}) \mapsto (z_1, \ldots, z_k) = \left( \prod_{i=0}^{k+2} y_i^{Q^1_i}, \ldots, \prod_{i=0}^{k+2} y_i^{Q^k_i} \right).
\ee
The Landau-Ginzburg superpotential $W : V \to \IC$ reads
\be
W = \sum_{i=0}^{k+2} y_i.
\ee
Choose local coordinates $y_0,y_1,y_2$ on the fiber $V_z = \pi^{-1}(z_1, \ldots, z_k)$, and use the projection map $\pi$ to rewrite the superpotential as
\be
W_z := W \big|_{V_{z}}  = G(y_0,y_1,y_2;z),
\ee
where $G(y_0,y_1,y_2;z)$ is a homogeneous Laurent polynomial in $(y_0,y_1,y_2)$ of degree 1.

Consider now the Landau-Ginzburg model on $V'_z = (\IC^*)^{3} \times (\IC^2)$, with superpotential
\be
W'_z = G(y_0,y_1,y_2;z) - w w'.
\ee
By Kn\"orrer periodicity, the category of B-branes in the Landau-Ginzburg model $(V_z,W_z)$ is equivalent to the category of B-branes in the Landau-Ginzburg model $(V'_z,W'_z)$ \cite{Or}. The ``periods'' of the Landau-Ginzburg model consists in integrals of the form
\be
\int \re^{G(y_0,y_1,y_2;z) - w w'} \dd \log y_0 \wedge \dd \log y_1 \wedge \dd \log y_2 \wedge \dd w \wedge \dd w'.
\ee

Since $y_0$ is a $\IC^*$-coordinate, we can define new coordinates $\tilde y_i = y_i / y_0$, $i=1,2$, and $\tilde w = w / y_0$. The superpotential becomes
\be
W'_z = y_0 (G(1,\tilde y_1, \tilde y_2; z) - \tilde w w'),
\ee
and the periods now take the form
\be
\int \re^{y_0 (G(1,\tilde y_1, \tilde y_2; z) - \tilde w w')} \dd y_0 \wedge \dd \log \tilde y_1 \wedge \dd \log \tilde y_2 \wedge \dd \tilde w \wedge \dd w'.
\ee
Note that $\dd \log y_0$ has become $\dd y_0$, due to the rescaling of $w$. As a result, we can ``integrate out'' $y_0$, and we obtain a delta function
\be
\delta(G(1,\tilde y_1, \tilde y_2; z) - \tilde w w').
\ee
In other words, the B-twisted Landau-Ginzburg model ``localizes'' on the B-twisted sigma model on the family of noncompact threefolds $Y \to \CM$ with fiber
\be
Y_z = \{ w w' = H(x,y;z) \} \subset (\IC^*)^2 \times (\IC^2),
\ee
where we redefined $x=\tilde y_1$, $y = \tilde y_2$, $w = \tilde w$ and $H(x,y;z) := G(1,x, y; z)$.

Note however that we have only shown equivalence of the period integrals. What one would need to show is the equivalence of the category of B-branes for both models, as in the first step involving Kn\"orrer periodicity. In other words, there should be an equivalence between the category of B-branes of the Landau-Ginzburg model $(V'_z, W'_z)$, which is generally understood as the category of matrix factorizations, and the derived category of coherent sheaves of $Y_z$. It would be very interesting to understand this relation better, in the spirit of the Landau-Ginzburg/Calabi-Yau correspondence which was derived in \cite{Or2}.

\section{Functions $f_i^{(0,3)}(x_1,x_2,x_3)$ for the genus $0$, $3$ hole amplitude}

We present in this Appendix the functions $f_i^{(0,3)}(x_1,x_2,x_3)$, $i=0,1,2$ entering into the expression for the amplitude $W^{(0,3)}$ \eqref{e:30exp}:

\begin{footnotesize}
\begin{align}
f_2^{(0,3)} =& \frac{1}{64\sigma(x_1)\sigma(x_2)\sigma(x_3)}(4 z (4 z (12 z x_3^3+(108 z+1) x_3^2+2 (54 z-1) x_3-3) x_2^3
+(4 z (108 z+1) x_3^3 \notag \\&-(216 z+5) x_3^2-6
   (54 z+1) x_3-108 z-1) x_2^2+2 (4 z (54 z-1) x_3^3-3 (54 z+1) x_3^2\notag\\&-2 (108 z+1) x_3-54 z+1) x_2 -12 z x_3^3
-(108 z+1)
   x_3^2+(2-108 z) x_3+3) x_1^3\notag\\&+(4 z (4 z (108 z+1) x_3^3 -(216 z+5) x_3^2-6 (54 z+1) x_3-108 z-1) x_2^3
\notag\\&+(-4 z (216
   z+5) x_3^3+9 (36 z+1) x_3^2+2 (270 z+7) x_3+216 z+5) x_2^2+(-24 z (54 z+1) x_3^3 \notag\\&+2 (270 z+7) x_3^2+4 (216 z+5) x_3+324
   z+6) x_2-4 z (108 z+1) x_3^3+(216 z+5) x_3^2
+108 z \notag\\&+6 (54 z+1) x_3+1) x_1^2+2 (4 z (4 z (54 z-1) x_3^3-3 (54 z+1)
   x_3^2-2 (108 z+1) x_3-54 z+1) x_2^3 \notag\\&+(-12 z (54 z+1) x_3^3+(270 z+7) x_3^2
+2 (216 z+5) x_3+162 z+3) x_2^2+(-8 z (108
   z+1) x_3^3\notag\\&+2 (216 z+5) x_3^2+12 (54 z+1) x_3+216 z+2) x_2+4 (1-54 z) z x_3^3+3 (54 z+1) x_3^2+54 z
\notag\\&+(216 z+2) x_3-1) x_1 +12 z
   x_3^3+108 z x_3^2+x_3^2+108 z x_3-2 x_3\notag\\&-4 z x_2^3 (12 z x_3^3+(108 z+1) x_3^2+2 (54 z-1) x_3-3)+x_2 (8 (1-54 z) z x_3^3+6
   (54 z+1) x_3^2\notag\\&+(432 z+4) x_3+108 z-2)+x_2^2 (-4 z (108 z+1) x_3^3+(216 z+5) x_3^2\notag\\&+6 (54 z+1) x_3+108 z+1)-3 ),
\end{align}
\begin{align}
f_1^{(0,3)}=&\frac{1}{192\sigma(x_1)\sigma(x_2)\sigma(x_3)} (-4 z (4 z (12 z (180 z+7) x_3^3+(2592 z^2-12 z-5) x_3^2-2 x_3+108 z+3) x_2^3\notag \\&+(4 z
   (2592 z^2-12 z-5) x_3^3+(1296 z^2+48 z+1) x_3^2+6 (648 z^2-12 z-1) x_3+1296 z^2-168 z-7) x_2^2\notag \\&-2
   (4 z x_3^3+(-1944 z^2+36 z+3) x_3^2+(-1944 z^2+360 z+14) x_3+324 z+11) x_2\notag \\&+12 z (36 z+1) x_3^3+(1296
   z^2-168 z-7) x_3^2-3 (144 z+5)-2 (324 z+11) x_3) x_1^3\notag \\&+(-4 z (4 z (2592 z^2-12 z-5) x_3^3+(1296
   z^2+48 z+1) x_3^2\notag \\&+6 (648 z^2-12 z-1) x_3+1296 z^2-168 z-7) x_2^3+(-4 z (1296 z^2+48 z+1) x_3^3\notag \\&+2
   (10368 z^2+492 z+5) x_3+9072 z^2+9 (36 z x_3+x_3){}^2+336 z+1) x_2^2\notag \\&+2 (12 z (-648 z^2+12 z+1)
   x_3^3+(10368 z^2+492 z+5) x_3^2+2 (7452 z^2+276 z+1) x_3\notag \\&+5832 z^2+108 z-3) x_2+4 z (-1296 z^2+168
   z+7) x_3^3+5184 z^2+(9072 z^2+336 z+1) x_3^2\notag \\&-24 z+6 (1944 z^2+36 z-1) x_3-7) x_1^2+2 (4 z (4 z
   x_3^3+(-1944 z^2+36 z+3) x_3^2\notag \\&+(-1944 z^2+360 z+14) x_3+324 z+11) x_2^3+(12 z (-648 z^2+12 z+1)
   x_3^3\notag \\&+(10368 z^2+492 z+5) x_3^2+2 (7452 z^2+276 z+1) x_3+5832 z^2+108 z-3) x_2^2\notag \\&-2 (4 z (972 z^2-180
   z-7) x_3^3-(7452 z^2+276 z+1) x_3^2+(6-5832 z^2) x_3-972 z^2+180 z+7) x_2\notag \\&+4 z (324 z+11) x_3^3+3
   (1944 z^2+36 z-1) x_3^2-324 z+2 (972 z^2-180 z-7) x_3-11) x_1+1728 z^2 x_3^3\notag \\&+60 z x_3^3+5184 z^2 x_3^2-24 z
   x_3^2-7 x_3^2-432 z-648 z x_3-22 x_3\notag \\&+4 z x_2^3 (-12 z (36 z+1) x_3^3+(-1296 z^2+168 z+7) x_3^2+(648 z+22) x_3+432
   z+15)\notag \\&+2 x_2 (4 z (324 z+11) x_3^3+3 (1944 z^2+36 z-1) x_3^2+2 (972 z^2-180 z-7) x_3-324 z-11)\notag \\&+x_2^2
   (4 z (-1296 z^2+168 z+7) x_3^3+(9072 z^2+336 z+1) x_3^2+6 (1944 z^2+36 z-1) x_3\notag \\&+5184 z^2-24
   z-7)-15),
\end{align}
\begin{align}
f_0^{(0,3)}=&\frac{1}{576\sigma(x_1)\sigma(x_2)\sigma(x_3)}(4 z (4 z (36 z (144 z^2+32 z+1) x_3^3+(7776 z^2+480 z+7) x_3^2\notag \\&+2 (4536 z^2+306
   z+5) x_3+3 (864 z^2+60 z+1)) x_2^3+(4 z (7776 z^2+480 z+7) x_3^3+(108864 z^3\notag \\&+7920 z^2+168
   z+1) x_3^2+6 (7776 z^3+1368 z^2+66 z+1) x_3+3888 z^2+276 z+5) x_2^2\notag \\&+2 (4 z (4536 z^2+306 z+5)
   x_3^3+3 (7776 z^3+1368 z^2+66 z+1) x_3^2\notag \\&+2 (4212 z^2+288 z+5) x_3+5184 z^2+378 z+7) x_2+12 z (864 z^2+60
   z+1) x_3^3+9 (24 z+1)^2\notag \\&+(3888 z^2+276 z+5) x_3^2+2 (5184 z^2+378 z+7) x_3) x_1^3\notag \\&+(4 z (4 z
   (7776 z^2+480 z+7) x_3^3+(108864 z^3+7920 z^2+168 z+1) x_3^2\notag \\&+6 (7776 z^3+1368 z^2+66 z+1) x_3+3888 z^2+276
   z+5) x_2^3+(124416 z^3+8496 z^2\notag \\&+4 (108864 z^3+7920 z^2+168 z+1) x_3^3 z+120 z+3 (186624 z^4+53568 z^3+3888
   z^2+108 z+1) x_3^2\notag \\&+(264384 z^3+20160 z^2+444 z+2) x_3-1) x_2^2+2 (46656 z^3+2376 z^2\notag \\&+12 (7776 z^3+1368
   z^2+66 z+1) x_3^3 z-54 z+(132192 z^3+10080 z^2+222 z+1) x_3^2\notag \\&+2 (81648 z^3+5292 z^2+60 z-1) x_3-3) x_2+4
   z (3888 z^2+276 z+5) x_3^3\notag \\&+(124416 z^3+8496 z^2+120 z-1) x_3^2-132 z+6 (15552 z^3+792 z^2-18 z-1)
   x_3-5) x_1^2\notag \\&+2 (4 z (4 z (4536 z^2+306 z+5) x_3^3+3 (7776 z^3+1368 z^2+66 z+1) x_3^2\notag \\&+2 (4212
   z^2+288 z+5) x_3+5184 z^2+378 z+7) x_2^3+(46656 z^3+2376 z^2\notag \\&+12 (7776 z^3+1368 z^2+66 z+1) x_3^3 z-54
   z+(132192 z^3+10080 z^2+222 z+1) x_3^2\notag \\&+2 (81648 z^3+5292 z^2+60 z-1) x_3-3) x_2^2+2 (4 z (4212
   z^2+288 z+5) x_3^3\notag \\&+(81648 z^3+5292 z^2+60 z-1) x_3^2+6 (5832 z^3+108 z^2-30 z-1) x_3-324 z^2-144 z-5)
   x_2\notag \\&+4 z (5184 z^2+378 z+7) x_3^3+2592 z^2+3 (15552 z^3+792 z^2-18 z-1) x_3^2-90 z\notag \\&-2 (324 z^2+144 z+5)
   x_3-7) x_1+20736 z^3 x_3^3+1728 z^2 x_3^3+36 z x_3^3+10368 z^2-132 z x_3^2\notag \\&-5 x_3^2+144 z+5184 z^2 x_3-180 z x_3-14 x_3+4 z x_2^3
   (12 z (864 z^2+60 z+1) x_3^3\notag \\&+(3888 z^2+276 z+5) x_3^2+2 (5184 z^2+378 z+7) x_3+9 (24
   z+1)^2)\notag \\&+x_2^2 (4 z (3888 z^2+276 z+5) x_3^3+(124416 z^3+8496 z^2+120 z-1) x_3^2\notag \\&+6 (15552 z^3+792
   z^2-18 z-1) x_3-132 z-5)+2 x_2 (4 z (5184 z^2+378 z+7) x_3^3\notag \\&+3 (15552 z^3+792 z^2-18 z-1) x_3^2-2
   (324 z^2+144 z+5) x_3+2592 z^2-90 z-7)-9).
\end{align}
\end{footnotesize}

\section{Conifold expansion}

In this Appendix we provide detailed calculations supporting the discussion of the open amplitudes near the conifold point in subsection \ref{s:con}.

Before computing the amplitudes, one needs to fix the open and closed mirror maps near the critical point $(w,p)=(0,0)$ in the moduli space. First, the closed mirror map can be easily obtained  by performing 
analytic continuation of the large radius periods to the conifold point $w=0$. One obtains
\bea
\label{cmm}
\nonumber
T_{con}&=&w+\frac{11 w^2}{18}+\frac{109 w^3}{243}+
\frac{9389 w^4}{26244}+\ldots ,
 \\
T^D_{con}&=&w \log (w) +\left(\frac{11 \log (w)}{18}+\frac{7}{12}\right) w^2+\left(\frac{109 \log (w)}{243}+\frac{877}{1458}\right) w^3+\ldots
\eea
The vanishing period $T_{con}$ at $w=0$ gives the 
closed flat coordinate at the conifold point, and the closed mirror map is obtained as usual by inverting the series.

The open mirror map is much more delicate.
As explained in \cite{BKMP}, it should be given by a linear combination of solutions of the extended Picard-Fuchs system. That is, by a linear combination of the constant solution, the
closed periods \eqref{cmm}, the 
solution (see \cite{BKMP}):
\be
u=\log(p-3)+\frac{1}{3}\log\left(\frac{w-1}{27}\right),
\ee
and the other relevant solution which is given by the disk amplitude.\footnote{By disk amplitude here we mean its completion with classical terms such that it is a solution of the extended Picard-Fuchs system.} At the critical point $(w,p)=(0,0)$, the latter reads
\bea
\label{aj}\nonumber
F^{(0,1)}&=&\ri \sqrt{3}\Biggl[
\frac{ p^2}{36}+\frac{11  p^3}{972}+
\frac{47  p^4}{11664} + \ldots
+w
\left(-\frac{1}{3}  \log (p) -\frac{ p}{54}+\frac{ p^2}{324}+
\frac{23  p^3}{13122} + \ldots
\right)\\&&
+w^2
   \left(-\frac{11}{54}  \log (p)-\frac{1}{2 p}+\frac{1}{2
   p^2}-\frac{29  p}{2916}
   +\frac{31  p^2}{17496} + \ldots\right) +\ldots\Biggr].
\eea
Therefore, generically, the open flat coordinate should be  of the form:
\bea
P=A\   u+B \  F^{(0,1)}+C\   T_{con}+D\   T^D_{con}+G.
\eea

We can directly set $B$ and $D$ to zero, as both  $F^{(0,1)}$ and 
$T^D_{con}$ contain a logarithm which would then introduce non-trivial monodromy in the physical disk amplitude.
We further decide to fix $A=-1$ and $G= \frac{4  \pi \ri}{3}$, for the following reasons. First, fixing $A$ just fixes the overall scale of the map. For instance, for $A=-1$ we get that
\be
P(p,w) = \frac{p}{3} + \frac{w}{3} + \ldots + C\  T_{con} + G - \frac{4 \pi \ri}{3}.
\ee
Then, we fix $G = \frac{4 \ri \pi}{3}$ to cancel the constant term in the $p,w$ expansion, as we want the flat coordinate to vanish at $(p,w)=(0,0)$. We then obtain
\be
P(p,w)=\frac{p}{3}+\frac{p^2}{18} + \ldots +\frac{w}{3}+\frac{w^2}{6} + \ldots + C\  T_{con},
\ee
and the inverse mirror map reads:
\bea
p&=&\nonumber
-(3C+1) T_{con}-\frac{1}{18} \left(27 C^2+18 C+1\right) T_{con}^2+ \ldots \\&&
+ \left(3+(3 C+1) T_{con}+ \frac{1}{18} \left(27 C^2+18 C+1\right) T_{con}^2+ \ldots \right) P+\ldots\eea
We did not however find any argument to fix the constant $C$ in the open mirror map. As a result, we are left with a one-parameter family of open mirror maps at the conifold, parameterized by $C$.

Now, as we already mentioned in footnote \ref{f:singular}, the conifold point is a singular limit for the
mirror curve, and one needs to choose appropriate coordinate on the resolution. To smooth out the singularity, as in \cite{EO} we  introduce the rescaled open flat coordinate:
\be
X=P \sqrt{T_{con}}.
\ee
We will then expand the open amplitudes in the flat coordinates $X$ and $T_{con}$.

Conifold amplitudes can be easily obtained 
with the method developed in section 2.
As for the orbifold point studied in section 3, basically we only need to input the flat coordinate $T_{con}$ in the functionals $W^{(g,h)}$, and expand the result in the flat coordinates $T_{con}$ and $X$ at the conifold point. We obtain the following results.

The disk amplitude at the conifold point reads:
\bea
\nonumber
 F^{(0,1)}&=&\Biggl[X+\frac{3 X^2}{4}+\ldots\biggr]
   T_{con}+
\biggl[-\frac{3 C X}{2}+\left(\frac{3
   C}{8}+\frac{1}{8}\right) X^2+\ldots\biggr] T_{con}^{3/2}\\&&
+\biggl[-\left(\frac{C}{4}+\frac{1}{72}+\frac{3 C^2 }{8}\right) X + \ldots\biggr] T_{con}^2+ \ldots .
\eea 
The annulus amplitude:
\bea
\nonumber F^{(0,2)}&=&\frac{3 X_1 X_2}{16}+\frac{9 X_1^2 X_2^2}{512}+\frac{9( X_1^3 X_2+ X_1 X_2^3)}{256}+\ldots\\&& \nonumber+
 \biggl[\frac{1}{64} \left(-9C+1\right) (X_2 X_1^2+ X_1 X_2^2)+\ldots\biggr]\sqrt{T_{con}}\\&&+ \biggl[\left(\frac{1}{24}  -\frac{C}{16}+\frac{9C^2}{32}\right)X_1
   X_2+\cdots \biggr]T_{con}+\ldots,
\eea
and the genus $1$, one-hole amplitude:
\bea
\nonumber 
F^{(1,1)}&=&\left(\frac{3}{32}X+\frac{15}{256}X^3+\cdots\right) \frac{1}{ T_{con}}+
   \biggl[
      \frac{
   1}{256} \left(1-45C\right) X^2
   +\ldots \biggr] \frac{ 1}{ \sqrt{T_{con}}    } \\&&+
\left(\frac{45  C^2}{256}-\frac{ C}{128}-\frac{77 }{2304}\right) X+\ldots+\nonumber
 \biggl[-\frac{15 C^3}{256}+\frac{C^2}{256}+\frac{77 C}{2304}+\frac{83}{6912}+
\\&&  + \left( \frac{135  C^2}{256} +\frac{595 C}{3072}+\frac{35}{9216} \right)X ^2
+\ldots\biggr]\sqrt{T_{con}}+\ldots.
\eea
These amplitudes  have indeed the expected leading 
behavior $T_{con}^{2-2g-h}$, as explained in \eqref{e:leading}.
However the subleading terms are not vanishing (for any value of $C$), 
and so there is no simple gap behavior.

\end{document}